\documentclass[Phys]{SciPost}

\pdfoutput=1
\usepackage[T1]{fontenc}
\usepackage{lmodern}
\usepackage{microtype}
\usepackage{amssymb,amsthm,mathtools}
\usepackage{bm}
\usepackage{tikz}
\usepackage{array}
\usepackage{booktabs}
\usepackage{longtable}
\usepackage{tabularx}
\usepackage{enumitem}
\usepackage{placeins}
\usepackage{orcidlink}
\hypersetup{
  colorlinks=true,
  linkcolor=red!50!black,
  citecolor=blue!50!black,
  urlcolor=blue!80!black
}
\definecolor{LeanVerifyBorder}{HTML}{1E6A8D}
\definecolor{LeanVerifyBackground}{HTML}{EFF8FC}
\newcommand{\leanverifiedbox}[1]{%
  \par\medskip
  \noindent\begingroup
  \setlength{\fboxsep}{7pt}%
  \setlength{\fboxrule}{0.8pt}%
  \fcolorbox{LeanVerifyBorder}{LeanVerifyBackground}{%
    \begin{minipage}{\dimexpr\linewidth-2\fboxsep-2\fboxrule\relax}
      \small
      \textcolor{LeanVerifyBorder}{\textbf{\checkmark\ Lean 4 machine verification}}\par\smallskip
      #1
    \end{minipage}%
  }%
  \endgroup
  \par\medskip
}
\usetikzlibrary{arrows.meta,calc,positioning,fit,backgrounds,decorations.pathreplacing}

\tikzset{
  smVertex/.style={circle, fill=black, inner sep=0pt, minimum size=4.0pt, label distance=1.5pt},
  smOpenVertex/.style={circle, draw=black, fill=black, inner sep=0pt, minimum size=4.0pt, label distance=1.5pt},
  smCoreVertex/.style={circle, draw=black, fill=black, inner sep=0pt, minimum size=4.0pt, label distance=1.5pt},
  smSmallLabel/.style={font=\scriptsize},
  smEdge/.style={thick},
  smDashedBlock/.style={draw=black, dashed, rounded corners=5pt, inner sep=4pt},
  smArrow/.style={-{Latex[length=2.0mm]}, thick},
  smMessage/.style={-{Latex[length=1.8mm]}, thick, dashed},
  smCountBox/.style={draw=black, rounded corners=4pt, inner sep=5pt, align=center, font=\scriptsize},
  gammaDictVertex/.style={
    circle,
    fill=black,
    draw=black,
    inner sep=0pt,
    minimum size=2.5pt
  },
  gammaDictEdge/.style={
    draw=black,
    line width=0.35pt,
    line cap=round
  },
  gammaDictName/.style={
    font=\scriptsize,
    align=center
  },
}

\newcommand{\CoreFourDictionaryGraph}[2]{%
\begin{tikzpicture}[
  baseline=(current bounding box.center),
  x=0.78cm,
  y=0.78cm
]
  \coordinate (v1) at (-0.72, 0.58);
  \coordinate (v2) at ( 0.72, 0.58);
  \coordinate (v3) at ( 0.72,-0.58);
  \coordinate (v4) at (-0.72,-0.58);

  \foreach \i/\j in {#2}{
    \draw[gammaDictEdge] (v\i)--(v\j);
  }

  \node[gammaDictVertex] at (v1) {};
  \node[gammaDictVertex] at (v2) {};
  \node[gammaDictVertex] at (v3) {};
  \node[gammaDictVertex] at (v4) {};

  \node[font=\scriptsize] at (-0.92, 0.78) {$a$};
  \node[font=\scriptsize] at ( 0.92, 0.78) {$b$};
  \node[font=\scriptsize] at ( 0.92,-0.79) {$c$};
  \node[font=\scriptsize] at (-0.92,-0.79) {$d$};
  \node[gammaDictName] at (0,-1.17) {#1};
\end{tikzpicture}%
}

\newcommand{\CoreFiveDictionaryGraph}[2]{%
\begin{tikzpicture}[
  baseline=(current bounding box.center),
  x=0.80cm,
  y=0.80cm
]
  \coordinate (v1) at ( 0.000, 0.950);
  \coordinate (v2) at ( 0.904, 0.294);
  \coordinate (v3) at ( 0.559,-0.769);
  \coordinate (v4) at (-0.559,-0.769);
  \coordinate (v5) at (-0.904, 0.294);

  \foreach \i/\j in {#2}{
    \draw[gammaDictEdge] (v\i)--(v\j);
  }

  \node[gammaDictVertex] at (v1) {};
  \node[gammaDictVertex] at (v2) {};
  \node[gammaDictVertex] at (v3) {};
  \node[gammaDictVertex] at (v4) {};
  \node[gammaDictVertex] at (v5) {};

  \node[font=\scriptsize] at ( 0.000, 1.155) {$a$};
  \node[font=\scriptsize] at ( 1.105, 0.365) {$b$};
  \node[font=\scriptsize] at ( 0.690,-0.985) {$c$};
  \node[font=\scriptsize] at (-0.690,-0.985) {$d$};
  \node[font=\scriptsize] at (-1.105, 0.365) {$e$};

  \node[font=\scriptsize] at (0,-1.235) {$#1$};
\end{tikzpicture}%
}

\setlist[itemize]{leftmargin=2em}
\setlist[enumerate]{leftmargin=2em}

\newtheorem{theorem}{Theorem}[section]
\newtheorem{proposition}[theorem]{Proposition}
\newtheorem{lemma}[theorem]{Lemma}

\newtheorem{definition}[theorem]{Definition}

\newcommand{\F}{\mathbb F}
\newcommand{\Fp}{\mathbb F_p}

\newcommand{\Zn}{\mathbb Z_n}
\newcommand{\TT}{\mathbb T_n}
\newcommand{\one}{\mathbf 1}
\newcommand{\ket}[1]{\left|#1\right\rangle}
\newcommand{\bra}[1]{\left\langle#1\right|}

\DeclareMathOperator{\rank}{rank}

\begin{document}

\begin{center}
{\Large \textbf{R\'enyi Multi-entropy for Stabilizer States: \\A Finite Field Gauss Sum Problem}}
\end{center}

\begin{center}
Bart{\l}omiej Czech\,\orcidlink{0000-0002-6368-3502}\textsuperscript{1},
Yichen Feng\,\orcidlink{0009-0003-2816-2714}\textsuperscript{1},
Jonathan Harper\,\orcidlink{0000-0003-4913-6470}\textsuperscript{2,3},
\\
Xin-Xiang~Ju\,\orcidlink{0000-0002-9905-3908}\textsuperscript{1,4},
Yang~Zhao\,\orcidlink{0009-0004-3492-1892}\textsuperscript{4}
\end{center}

\begin{center}
{\bf 1} Chen-Ning Yang Institute for Advanced Study, Tsinghua University, Beijing, China
\\
{\bf 2} Department of Physics, Case Western Reserve University,
\\10900 Euclid Ave, Cleveland, OH 44106, USA
\\
{\bf 3} Center for Gravitational Physics and Quantum Information,
\\Yukawa Institute for Theoretical Physics, Kyoto University,
\\Kitashirakawa Oiwakecho, Sakyo-ku, Kyoto 606-8502, Japan
\\
{\bf 4} School of Physical Sciences, University of Chinese Academy of Sciences, Beijing, China
\\
Email:
{\small \sf \href{mailto:bartlomiej.czech@gmail.com}{bartlomiej.czech@gmail.com}},
{\small \sf \href{mailto:fyc23@mails.tsinghua.edu.cn}{fyc23@mails.tsinghua.edu.cn}},
{\small \sf \href{mailto:jonathan.harper@case.edu}{jonathan.harper@case.edu}},
{\small \sf \href{mailto:juxinxiang21@mails.ucas.ac.cn}{juxinxiang21@mails.ucas.ac.cn}},
{\small \sf \href{mailto:zhaoyang20a@mails.ucas.ac.cn}{zhaoyang20a@mails.ucas.ac.cn}}
\end{center}

\section*{Abstract}
R\'enyi multi-entropy probes multi-partite structure beyond bipartite spectra, but its behavior in general quantum states and the entanglement structure it captures remain poorly understood. We develop a finite field Gauss sum framework for calculating it in stabilizer states. Within this framework, the replica contraction reduces to a quadratic Gauss sum if and only if the state is locally unitarily equivalent to a stabilizer state. We obtain analytic expressions for the R\'enyi multi-entropy of all tree graph states and large families of loop graph states. In contrast to their flat bipartite R\'enyi spectra, stabilizer states generically have nonflat R\'enyi multi-entropy spectra. 

\vspace{10pt}
\noindent\rule{\textwidth}{1pt}
\tableofcontents\thispagestyle{fancy}
\noindent\rule{\textwidth}{1pt}
\vspace{10pt}

\section{Introduction}
Much of our understanding of quantum entanglement rests on dividing a system into two parts. For a pure state, the spectrum of either reduced density matrix fully characterizes the entanglement across that bi-division, up to local unitaries. When the state is shared by more than two parties (subsystems), however, the entanglement among them cannot be classified by data from all bi-divisions~\cite{DurVidalCirac2000}. This motivates a study of quantities, which are sensitive to the way entanglement is shared among more than two parties. The present paper focuses on one class of such quantities, called multi-R\'enyi entropies.

Ordinary R\'enyi entropies \cite{renyi1961entropy} provide a natural starting point for defining quantities sensitive to entanglement across multiple parties (multi-partitions). Each moment \(\operatorname{Tr}\rho_A^n\) is a closed contraction of copies of a pure state and its conjugate. For a multipartite state, a natural generalizations of this construction is to take a number of labeled replicas and set a rule for which bra-replica index is to be contracted with which ket-replica index. As explained in \cite{LeiferLindenWinter2004} and elsewhere, this specification can be encapsulated by a set of permutations acting on the replicas of the state, one permutation per subsystem. In the case of the R\'enyi multi-entropies, which are the focus of this paper, the permutations generate a $(q-1)$-fold product of $\mathbb{Z}_n$'s, where $q$ is the number of constituent subsystems. Regardless of the choice of permutations, the scalar defined by such a pattern of contractions is invariant under local changes of basis and, for this reason, bears the name \emph{multi-invariant}. Replica-symmetric families of multi-invariants have been studied in holographic duality \cite{GaddeHarperKrishna2025}, and more recent work has investigated their properties for stabilizer states \cite{AkellaGaddePandey2026,CzechFengWuXie2026}. 

R\'enyi multi-entropies are a distinguished family
of multi-invariants, in which the replicas of the state are associated with group elements of products of $\mathbb{Z}_n$. They reduce to ordinary R\'enyi entropies for two parties and treat all parties symmetrically. The original motivation
for studying R\'enyi multi-entropies was holographic: an analytic continuation to \(n=1\) was proposed to
compute the area of a network of minimal surfaces in the dual bulk geometry~\cite{GaddeKrishnaSharma2022}. The integer-replica definition, however, does not require a gravitational dual. It is natural to ask: what aspects of multipartite entanglement are reflected in the  dependence of the R\'enyi multi-entropies on the replica index? To begin to address this question, we compute the multi-R\'enyis exactly in a well-studied family of quantum states:


\paragraph{Stabilizer states} provide a particularly clean testbed for this problem.
While they admit a compact classical description, the structure of their entanglement across multi-partitions exhibits non-trivial features, which we help to elucidate in this paper. A distinctive feature of stabilizer states is that their bipartite entanglement spectra are flat \cite{fattal2004stabilizer} and, consequently, their R\'enyi entropies are constant (independent of the replica index). This is not the case for the multi-R\'enyi entropies.

\paragraph{Technique: Gauss sums}
We recast the calculation of R\'enyi multi-entropies for stabilizer states as a finite field Gauss sum problem. In a suitable local basis, their nonzero amplitudes have equal magnitude, with phases determined by a quadratic polynomial over a finite field. Replica gluing imposes linear conditions on the finite field elements, leading to a quadratic Gauss sum over the space of solutions to certain linear conditions. In the end, our methods reduce the calculation of an index contraction to a linear algebra problem over a finite field.

\paragraph{Results}
Analytic results for R\'enyi multi-entropies at arbitrary integer R\'enyi
index \(n\) are only known for limited classes of states
\cite{BerthiereGaudin2026,AkellaGaddePandey2026,ChangGoswamiIizukaMiyata2026},
with GHZ states and their generalizations providing the standard examples
\cite{GaddeKrishnaSharma2022,PeningtonWalterWitteveen2023}.
For more general states, explicit calculations have largely been confined
to low replica indices
\cite{HarperTakayanagiTsuda2024,HarperMollabashiTakayanagiTasuki2025,IizukaLin2026,CamargoNishida2026}. In a companion paper \cite{CzechFengHarperJu2026}, we study a broader class of multi-invariants (beyond the multi-R\'enyi entropies), but restricted to a narrow class of stabilizer states (equation~(\ref{eq:twoblockstate}) and close cousins thereof).

In this paper, we organize the information about a stabilizer state in the form of a ``core graph;'' see Section~\ref{sec:coregraphrep}. In a core graph, a vertex is associated not to an individual qubit but a GHZ-like block. Applying the Gauss sum framework to the core graph representation of a stabilizer state substantially extends the range in which arbitrary-\(n\) analytic results can be obtained:

\begin{itemize}
\item  Treating each qubit as a separate party, we derive a closed formula
  for the R\'enyi multientropy of every tree graph state, valid for every
  integer \(n\ge2\) (Section~\ref{sec:trees}).
  We also obtain analytic expressions for large families of graph
  states containing loops, including odd cycles of arbitrary length
  for odd \(n\). A decomposition rule further extends the odd-\(n\)
  results to larger graphs (Section~\ref{sec:loops}).
\item We also extend the framework to coarser multipartitions (more than one qubit per party) and prime-dimensional qudits (Section~\ref{sec:generalization-scope}).
\end{itemize}

\noindent This paper is a direct continuation of our accompanying work \cite{CzechFengHarperJu2026}.

\paragraph{Notes and disclaimers}
A considerable proportion of this work was assisted by GPT-6 Astra and earlier models. We include a detailed statement of where and how it was used in the Discussion. All analytical proofs in this paper have been verified and simplified by us, the human authors. We have not attempted to verify enumerative calculations such as equation~(\ref{eq:cubic-small-n}). Several key results have been formalized and machine-verified in Lean~4~\cite{deMoura2021Lean4,mathlib2020}. The Lean source code and build instructions are publicly available in our GitHub repository~\cite{StabilizerRenyiLean2026}.

As we were completing this paper, we learned of concurrent work on the same subject by Abhijit Gadde and Sriram Akella. We thank these authors for synchronizing the arXiv submissions of our papers.

\section{Graph state representatives and finite field assignments}
\label{sec:calculus}

This section establishes finite field Gauss sums as the mathematical tool for computing R\'enyi multi-entropy for stabilizer states. Section \ref{subsec:2.1} briefly reviews R\'enyi multi-entropy and the graph state description of stabilizer states. In Section \ref{sec:coregraphrep}, we introduce the core-graph representation, which simplifies this description and organizes the states into families sharing a common core graph. Section \ref{subsec:finite-fourier-gauss} associates each such family with a single parameterized Gauss sum problem. Solving this problem yields the R\'enyi multi-entropy for the entire family, without requiring a separate calculation for each graph state.

As examples, Section \ref{subsec:worked-examples} gives complete evaluations for two of the simplest core-graph families beyond GHZ states. Finally, Section \ref{subsec:2.5} shows that local complementation of graph states is implemented by a Fourier identity in the finite field formulation. This provides a direct consistency check: the Gauss sum description preserves the local-unitary invariance of R\'enyi multi-entropy.

\subsection{The R\'enyi multi-entropy and the graph state representation of stabilizer states}
\label{subsec:2.1}

We begin with a brief introduction to R\'enyi multi-entropy and stabilizer states, the objects of our study.

\subsubsection{R\'enyi multi-entropy}

Consider a normalized pure state \(\lvert\Psi\rangle\) of \(q\) parties \(A_1,A_2,\ldots,A_{q-1}\), with Hilbert-space dimensions \(d_0,d_1,\ldots,d_{q-1}\). We also use the label \(A_q\equiv A_0\). In an orthonormal product basis,
\begin{equation}
  \lvert\Psi\rangle
  =\sum_{i_0=1}^{d_0}\cdots\sum_{i_{q-1}=1}^{d_{q-1}}
  \psi_{i_0\cdots i_{q-1}}
  \lvert e^{i_0}\rangle\otimes\cdots\otimes\lvert e^{i_{q-1}}\rangle,
\end{equation}
and
\begin{equation}
  \langle\Psi\rvert
  =\sum_{i_0=1}^{d_0}\cdots\sum_{i_{q-1}=1}^{d_{q-1}}
  \bar\psi^{i_0\cdots i_{q-1}}
  \langle e_{i_0}\rvert\otimes\cdots\otimes\langle e_{i_{q-1}}\rvert.
\end{equation}
The density matrix is \(\rho=\lvert\Psi\rangle\langle\Psi\rvert\).

For permutations \(g_0,\ldots,g_{q-1}\in S_s\) of \(s\) replicas, the corresponding \emph{multi-invariant} is
\begin{equation}
\label{inv-def}
\begin{aligned}
  \mathcal E(g_0,\ldots,g_{q-1})
  &={}
  \left(\prod_{t=1}^{s}\psi_{i_0^{(t)}\cdots i_{q-1}^{(t)}}\right)
  \left(\prod_{t=1}^{s}\bar\psi^{j_0^{(t)}\cdots j_{q-1}^{(t)}}\right)
  \prod_{\mu=0}^{q-1}\delta^{\vec i_\mu}_{g_\mu\cdot\vec j_\mu},
  \\
  \delta^{\vec i_\mu}_{g_\mu\cdot\vec j_\mu}
  &:=\prod_{t=1}^{s}
  \delta^{i_\mu^{(t)}}_{j_\mu^{(g_\mu\cdot t)}}.
\end{aligned}
\end{equation}
Repeated tensor indices are summed. Each permutation specifies how the indices of one party are contracted between the ket and bra replicas.

An important class of multi-invariants, and the primary focus of this article, comprises replica contractions that define the \(q\)-party R\'enyi-\(n\) multi-entropy. These contractions have replica symmetry group \(G=(\mathbb Z_n)^{q-1}\). We label the parties \(A_0,A_1,\ldots,A_{q-1}\), with \(A_q\equiv A_0\), and choose \(A_0\) as the ``reference'', setting \(g_0=e\). The remaining permutations \(g_1,\ldots,g_{q-1}\) are independent commuting generators of \(G\), each acting as a cyclic shift along one coordinate direction. The contraction can therefore be represented on a \((q-1)\)-dimensional toric lattice of period \(n\), with each site corresponding to one copy of the density matrix \(\rho\). Formally, we define:

\begin{definition}[integer-replica contraction for the R\'enyi multi-entropy]
\label{def:integer-replica-contraction}
Let \(q\ge2\) and \(n\ge2\) be integers, and let
\(\psi_{i_0\cdots i_{q-1}}\) be a normalized \(q\)-party pure-state tensor.
With \(A_0\) as reference, set
\begin{equation}
    \TT^{(q)}:=(\Zn)^{q-1},
  \qquad
  e_0=0,
\end{equation}
where \(e_1,\ldots,e_{q-1}\) are the standard unit vectors and all coordinates are understood modulo \(n\). Represent the permutations by \(g_\mu\cdot u=u-e_\mu\) for \(\mu=0,\ldots,q-1\). The integer-replica contraction is the special case of Eq.~\eqref{inv-def} with these permutations and \(s=n^{q-1}\):
\begin{equation}
  \mathcal Z_n(\psi)=
  \sum_{\{i_\mu(u)\}}
  \prod_{u\in\TT^{(q)}}
  \psi_{i_0(u),i_1(u),\ldots,i_{q-1}(u)}
  \overline{
    \psi_{i_0(u),i_1(u+e_1),\ldots,i_{q-1}(u+e_{q-1})}
  }.
  \label{eq:general-toric-contraction}
\end{equation}
\end{definition}

We write \(\mathcal Z_n(\psi)\) for the replica trace contraction in Eq.~\eqref{eq:general-toric-contraction}, which contains \(n^{q-1}\) copies of \(\psi\) and the same number of its complex conjugate. For a normalized state, the R\'enyi multi-entropy used in this paper is
\begin{equation}
  S_n(\psi)
  =
  -\frac{\log\mathcal Z_n(\psi)}{(n-1)n^{q-2}},
  \qquad \mathcal Z_n(\psi)>0.
  \label{eq:entropy-normalization}
\end{equation}
Our convention differs from the unrescaled expression \(-\log\mathcal Z_n/(n-1)\) by a factor \(n^{-(q-2)}\). For \(q=2\), \(\mathcal Z_n=\operatorname{Tr}\rho_{A_1}^{\,n}\), where \(\rho_{A_1}=\operatorname{Tr}_{A_0}\rho\), and \(S_n\) reduces to the ordinary R\'enyi entropy.

For real wavefunction amplitudes, as used in Sections \ref{sec:calculus}--\ref{sec:loops}, complex conjugation leaves the coefficients unchanged. We retain the overline in Eq. \eqref{eq:general-toric-contraction} because the definition applies to general complex amplitudes, including those considered in Section~\ref{subsec:prime-qudits}.



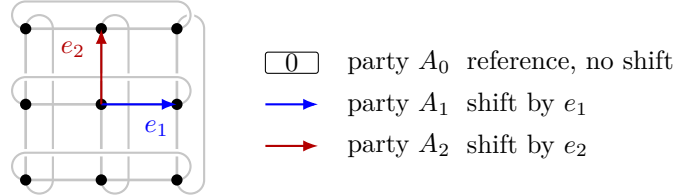
\begin{figure}[!h]
\centering
\begin{tikzpicture}[
  every node/.style={font=\small},
  site/.style={circle, fill=black, inner sep=1.5pt},
  arr/.style={-{Latex[length=2mm]}, thick},
  wrap/.style={thick, line cap=round}
]

\def\wrapradius{0.18}

\foreach \x in {0,1,2} {
  \foreach \y in {0,1,2} {
    \coordinate (p\x\y) at (\x,\y);
  }
}

\foreach \k in {0,1,2} {
  \draw[gray!45, line width=1pt] (\k,0) -- (\k,2);
  \draw[gray!45, line width=1pt] (0,\k) -- (2,\k);
}

\foreach \x in {0,1,2} {
  \draw[wrap, gray!45]
    (p\x2)
    arc[
      start angle=180,
      end angle=0,
      radius=\wrapradius
    ]
    -- ({\x+2*\wrapradius},0)
    arc[
      start angle=0,
      end angle=-180,
      radius=\wrapradius
    ];
}

\foreach \y in {0,1,2} {
  \draw[
    wrap,
    gray!45,
    preaction={draw=white, line width=2.8pt}
  ]
    (p2\y)
    arc[
      start angle=-90,
      end angle=90,
      radius=\wrapradius
    ]
    -- (0,{\y+2*\wrapradius})
    arc[
      start angle=90,
      end angle=270,
      radius=\wrapradius
    ];
}

\foreach \x in {0,1,2} {
  \foreach \y in {0,1,2} {
    \node[site] at (p\x\y) {};
  }
}

\draw[arr, blue]
  (p11) -- (p21)
  node[pos=0.73, below=2pt] {\(e_1\)};

\draw[arr, red!70!black]
  (p11) -- (p12)
  node[pos=0.73, left=2pt] {\(e_2\)};

\node[
  draw,
  rounded corners=1pt,
  minimum width=0.70cm,
  minimum height=0.28cm,
  inner sep=0pt
] at (3.52,1.55) {\(0\)};

\node[anchor=west] at (4.12,1.55) {party \(A_0\)};
\node[anchor=west] at (5.7,1.55) {reference, no shift};

\draw[arr, blue]
  (3.17,1.00) -- (3.87,1.00);

\node[anchor=west] at (4.12,1.00) {party \(A_1\)};
\node[anchor=west] at (5.7,1.00) {shift by \(e_1\)};

\draw[arr, red!70!black]
  (3.17,0.45) -- (3.87,0.45);

\node[anchor=west] at (4.12,0.45) {party \(A_2\)};
\node[anchor=west] at (5.7,0.45) {shift by \(e_2\)};

\end{tikzpicture}
\caption{Replica contractions for $q=3$ parties. The reference party \(A_0\) is contracted without a shift, while the \(A_1\)- and \(A_2\)-indices are contracted with $\mathbb{Z}_n$-shifted replicas. The full replica group is $\mathbb{Z}_n^{q-1}$, here shown for $n=3$.}
\label{fig:replica-torus-gluing}
\end{figure}
\FloatBarrier

\begin{lemma}[partition-local unitary invariance]
\label{lem:partition-local-unitary-invariance}
\label{lem:one-body-phases}
For a fixed partition and reference party, the contraction \(\mathcal Z_n\) in \eqref{eq:general-toric-contraction} is invariant under \(U_1\otimes\cdots\otimes U_q\), where each \(U_\mu\) acts on the corresponding party of that partition.
\end{lemma}

\begin{proof}
After applying \(U_1\otimes\cdots\otimes U_q\), each party contributes one matrix \(U_\mu\) from the ket coefficient and one conjugate matrix \(U_\mu^\dagger\) from the shifted bra coefficient.  For the reference party these matrices are contracted at the same torus site.  For a non-reference party \(\mu\), the bra index is shifted by \(e_\mu\), but translation by \(e_\mu\) is a bijection of \(\TT^{(q)}\).  Relabeling the torus site along each closed replica link therefore pairs every \(U_\mu\) with \(U_\mu^\dagger\), and unitarity collapses each pair to the identity.
\end{proof}

\subsubsection{Graph states}
In this paper we calculate the R\'enyi multi-entropy for \emph{stabilizer states}. It is well-known that every stabilizer state is locally Clifford equivalent to a graph state \cite{hein2006graph,vandennest2004graph}. Therefore, in what follows our treatment of stabilizer states is phrased in terms of graph states. 

Given a graph \(G_{\rm phys}=(V_{\rm phys},E_{\rm phys})\) with \(N\) vertices, a graph state is prepared as follows. Starting from the state \(\ket{+}^{\otimes N}\), act with the gate $CZ = (-1)^{(1-Z) \otimes (1-Z)/4}$ on each pair of qubits whose corresponding vertices are connected by an edge:
\begin{equation}
    \ket{G_{\rm phys}}=\prod_{\{i,j\}\in E_{\rm phys}}CZ_{ij}\ket{+}^{\otimes N}
\end{equation}
Importantly for us, the graph state can be equivalently written as:
\begin{equation}
    \ket{G_{\rm phys}}=
  2^{-N/2}\sum_{x\,\in\,\F_2^N}
  (-1)^{D_{\rm phys}(x)}\ket{x_1\cdots x_N}
  \quad {\rm with}~~
  D_{\rm phys}(x)=\sum_{\{i,j\}\,\in\, E_{\rm phys}}x_ix_j.
\end{equation}
Here the kets $\ket{x_1\cdots x_N}$ are $(-1)^{x_i}$-eigenstates of $Z_i$.  As such, we interpret the $x_i$ as elements of the binary field \(\mathbb{F}_2\). The exponent \(D_{\rm phys}(x)\) is fully determined by the graph \(G\) and is a \emph{quadratic Boolean polynomial}: each \(CZ\) edge contributes exactly one bilinear monomial \(x_i x_j\).  This quadratic structure allows us to express the replica contraction as a finite field Gauss sum.


\subsection{Core-graph representation of stabilizer states}
\label{sec:coregraphrep}
We introduce a \emph{core-graph representation}, which organizes stabilizer states into families of coupled GHZ blocks. In a suitable local basis, all qubits within a block share the same binary label, so each block is represented by a single core vertex, while the edges record the quadratic couplings between blocks. Varying the block sizes gives a family of states with the same core graph but different numbers of physical qubits. As we show in Section~\ref{subsec:finite-fourier-gauss}, their R\'enyi multientropies can be obtained from a single calculation with the block sizes as parameters.


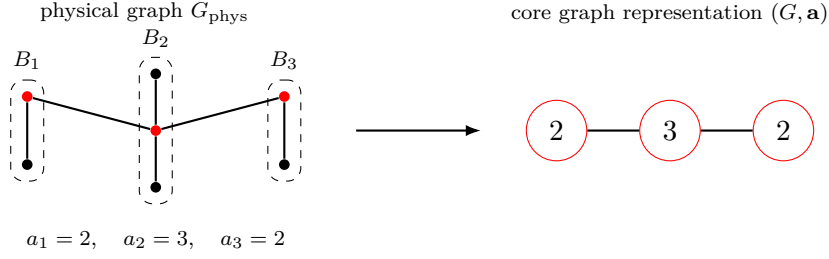
\begin{figure}[t]
\centering
\begin{tikzpicture}[scale=1.0]

\node[smSmallLabel] at (1.55,1.55) {physical graph $G_{\rm phys}$};

\node[smVertex,color=red] (a1) at (0,0.45) {};
\node[smVertex] (a2) at (0,-0.45) {};

\node[smVertex] (b1) at (1.7,0.75) {};
\node[smVertex,color=red] (b2) at (1.7,0) {};
\node[smVertex] (b3) at (1.7,-0.75) {};

\node[smVertex,color=red] (c1) at (3.4,0.45) {};
\node[smVertex] (c2) at (3.4,-0.45) {};

\draw[smEdge] (a1)--(a2);
\draw[smEdge] (a1)--(b2);
\draw[smEdge] (b1)--(b2);
\draw[smEdge] (b2)--(b3);
\draw[smEdge] (b2)--(c1);
\draw[smEdge] (c1)--(c2);

\node[smDashedBlock, fit=(a1)(a2), label=above:{\scriptsize \(B_1\)}] {};
\node[smDashedBlock, fit=(b1)(b2)(b3), label=above:{\scriptsize \(B_2\)}] {};
\node[smDashedBlock, fit=(c1)(c2), label=above:{\scriptsize \(B_3\)}] {};

\node[smSmallLabel, align=center] at (1.7,-1.45)
{\(a_1=2,\quad a_2=3,\quad a_3=2\)};

\draw[smArrow] (4.35,0) -- (6,0);

\node[smSmallLabel] at (8.5,1.55) {core graph representation $(G,\mathbf a)$};

\node[draw, circle, color=red, minimum size=8mm, inner sep=0pt] (A) at (7,0,0) {\color{black} 2};
\node[draw, circle, color=red, minimum size=8mm, inner sep=0pt] (B) at (8.5,0,0) {\color{black} 3};
\node[draw, circle, color=red, minimum size=8mm, inner sep=0pt] (C) at (10,0,0) {\color{black} 2};

\draw[smEdge] (A)--(B);
\draw[smEdge] (B)--(C);

\end{tikzpicture}
\caption{Core-graph representation of a seven-qubit graph state $\ket{G_{\rm phys}}$. The vertices are partitioned into three star-shaped blocks $B_1,B_2,B_3$ of sizes $(a_1,a_2,a_3)=(2,3,2)$. Their red central vertices and the edges connecting them form the core graph $G$, shown on the right together with the block sizes. By Proposition~\ref{prop:coregraphstate-physicalgraphstate}, $\ket{G_{\rm phys}}$ is local-Clifford equivalent to the corresponding core-graph state $\ket{\Psi_{G,\mathbf{a}}}$.}
\label{fig:ghz-block-compression}
\end{figure}

We illustrate how a core-graph representation works for a given stabilizer state through an example first. Consider a 7-qubit graph state $\ket{G_{\rm phys}}$ shown in Figure \ref{fig:ghz-block-compression}. The graph consists of three star subgraphs $B_1, B_2, B_3$ centered at the red qubits $c_1, c_2, c_3$ connected along a path, with block sizes $a_1=2, a_2=3, a_3=2$. Specifically, $c_1$ is attached to one leaf, $c_2$ to two leaves, and $c_3$ to one leaf.

Applying a single-qubit Hadamard gate $H_r$ to every leaf qubit $r$ transforms this seven-qubit state into
\begin{equation}
\begin{aligned}\label{eq: core3example}
  \left(\prod_{r\in\text{leaves}} H_r\right)\ket{G_{\rm phys}}
  &\propto
  \sum_{x_1,x_2,x_3\,\in\,\F_2} (-1)^{x_1 x_2 + x_2 x_3} \ket{x_1^2 \, x_2^3 \, x_3^2}\\
  &=
  \ket{00\,000\,00} + \ket{00\,000\,11} + \ket{00\,111\,00} - \ket{00\,111\,11}\\
  &\quad + \ket{11\,000\,00} + \ket{11\,000\,11} - \ket{11\,111\,00} + \ket{11\,111\,11}.
\end{aligned}
\end{equation}
Notice that this state is simply a three-qubit chain graph state $2^{-3/2}\sum_x (-1)^{x_1 x_2 + x_2 x_3}\ket{\bar{x}_1 \bar{x}_2 \bar{x}_3}$, with each core variable $x_v$ encoded into an $a_v$-qubit repetition code:
\begin{equation}
    V_v:\ket{\bar{0}}\mapsto\ket{0^{a_v}}:=\ket{0}^{\otimes a_v},
  \quad
  \ket{\bar{1}}\mapsto\ket{1^{a_v}}:=\ket{1}^{\otimes a_v},
    \quad
  \forall v\in V.
\end{equation}
Therefore, we call this simpler encoded description the \emph{core-graph state} $\ket{\Psi_{G,\mathbf a}}$. As shown in Figure \ref{fig:ghz-block-compression}, for any graph state $\ket{G_{\rm phys}}$, we collapse it to a smaller core graph $G=(V,E)$ by pruning all leaves. For each non-leaf vertex $v\in V$, we set its block size
\begin{equation}
  a_v = 1 + \#\{\text{pruned leaves attached to } v\}.
\end{equation}
We then define the corresponding \emph{core-graph state} $\ket{\Psi_{G,\mathbf a}}$ in the computational basis by
\begin{equation}\label{eq:core-state}
\ket{\Psi_{G,\mathbf a}}
:=2^{-k/2}\sum_{x\,\in\,\F_2^{|V|}}(-1)^{D_G(x)+\ell(x)}
\bigotimes_{v\,\in\, V}\ket{x_v^{a_v}},
\qquad
D_G(x)=\sum_{(u,v)\,\in\, E}x_ux_v ,
\end{equation}
where $k:=|V|$ and \(\ell(x):=\sum_v l_v x_v\) is an optional linear phase that can be removed by local unitaries. By construction, a core-graph can LC-equivalently represent \textit{any} stabilizer state. The following proposition explicitly shows the LC-equivalency.
\begin{proposition}[Local-Clifford equivalence to core-graph states]\label{prop:coregraphstate-physicalgraphstate}
Every graph state $\ket{G_{\rm phys}}$ is local-Clifford equivalent to its core-graph state $\ket{\Psi_{G,\mathbf a}}$ via Hadamard gates on leaves and Pauli-$Z$ gates on centers:
\begin{equation}
  \ket{G_{\rm phys}}
  =
  \left(\prod_{v\in V} Z_{c_v}^{l_v}\right)
  \left(\prod_{r\in \text{leaves}} H_r\right)
  \ket{\Psi_{G,\mathbf a}},
\end{equation}
where $c_v$ is the center qubit of block $v$.
\end{proposition}
\begin{proof}
On each block $v$, the leaf stabilizers $Z_{c_v}Z_r$ are mapped by the leaf Hadamards $H_r$ to $X_r Z_{c_v}$, matching the physical graph state generators for the leaves. The encoded generator $K_v = (-1)^{l_v}(\prod_{j\in B_v} X_j)(\prod_{u\in N_G(v)} Z_{c_u})$ is mapped by $\prod_{r\in L_v} H_r$ and $Z_{c_v}^{l_v}$ to $X_{c_v}(\prod_{r\in L_v} Z_r)(\prod_{u\in N_G(v)} Z_{c_u})$, which is the canonical physical graph state generator for the center $c_v$.
\end{proof}
For a fixed graph state, finding its core-graph representation is straightforward: One simply prunes all leaves and attaches to every non-leaf vertex $v$ a quantity $a_v$, set to the number of pruned vertices plus one. We point out, however, that distinct graph states in the same LC orbit have distinct core-graph representations, some of which may be far simpler than others. If one had a way of efficiently finding the LC-equivalent graph representation with the smallest core-graph, the computations presented in this paper would be amenable to a further simplification. Unfortunately, scanning over the orbit of LC-equivalent graph states is known to be NP-hard \cite{Adcock2020,Sharma:2025kkb}.

From now on, we work directly with core-graph states, but our results apply to all stabilizer states.

\subsection{R\'enyi multi-entropy as a finite field Gauss sum problem}
\label{subsec:finite-fourier-gauss}

The core-graph representation introduced above gives a simple starting point for the replica calculation. Its nonzero amplitudes all have the same magnitude, and their signs are determined by the quadratic polynomial $D_G$. We now show how this representation turns the calculation of R\'enyi multi-entropy into a finite field Gauss sum problem, and explain the basic summations used to solve it.

Let $G$ have $k$ core vertices with GHZ-block sizes $a_v$. We treat each physical qubit as a separate party, so that $q=\sum_{v=1}^k a_v$. For an integer $n\geq 2$, the R\'enyi multi-entropy is
\begin{equation}
  S_n(G)
  =
  -\frac{\log \mathcal Z_n(G)}{(n-1)n^{q-2}},
  \qquad \mathcal Z_n(G)>0,
  \label{eq:entropy-normalizationII}
\end{equation}
where $\mathcal Z_n(G)$ is the replica contraction of the normalized state defined in Section \ref{subsec:2.1}. It contains $n^{q-1}$ copies of the wavefunction and the same number of its complex conjugate. For $q=2$, this definition reduces to the ordinary R\'enyi entropy.

Choose \(A_0\) as the reference party and label its GHZ block by \(v=1\). We group the replica-torus coordinates by GHZ block, writing each site as \((\mathbf r,\mathbf y,\mathbf z,\ldots)\). The vector \(\mathbf r\) contains the \(a_1-1\) coordinates associated with the non-reference qubits in the first block. Each subsequent block \(v\) contributes \(a_v\) coordinates, grouped into \(\mathbf y,\mathbf z,\ldots\) in block order. All coordinate components take values in \(\Zn\).

An {\bf assignment} specifies a ket and a bra bit string from the core-graph expansion at every site of the replica torus. It is \emph{valid} if it satisfies the gluing conditions: all bit labels paired by the replica contraction must agree. Each valid assignment gives one nonzero term in the contraction, equal to \(+1\) or \(-1\) after removing the common normalization factor.

For each GHZ block $v=1,\ldots,k$, introduce the ket and bra {\bf assignment fields}:
\begin{equation}
\label{defassignmentfields}
    A_v^{ket},A_v^{bra}:
  \TT^{(q)}=(\Zn)^{q-1}
  \longrightarrow \F_2.
\end{equation}
An assignment field specifies the binary label of its block at every replica site $(\mathbf r,\mathbf y,\mathbf z,\ldots)$. All physical qubits in block $v$ carry the label $A_v^{ket}(\mathbf r,\mathbf y,\mathbf z,\ldots)$ in the ket and $A_v^{bra}(\mathbf r,\mathbf y,\mathbf z,\ldots)$ in the bra. In the two-block example below, the fields of the first block are written as $A^{ket},A^{bra}$, and those of the second block as $B^{ket},B^{bra}$.

We collect the assignment fields into
\begin{equation}
    \mathbf A^{ket}=(A_1^{ket},\ldots,A_k^{ket}),
  \qquad
  \mathbf A^{bra}=(A_1^{bra},\ldots,A_k^{bra}).
\end{equation}
Thus $\mathbf A^{ket}(\mathbf r,\mathbf y,\mathbf z,\ldots)$ and $\mathbf A^{bra}(\mathbf r,\mathbf y,\mathbf z,\ldots)$ are the binary strings labeling the ket and bra basis states at that site.

If physical qubit $\mu$ belongs to block $v(\mu)$, the replica contraction identifies its bra label with its ket label at the appropriate shifted site:
\begin{equation}
  A_{v(\mu)}^{bra}(\mathbf r,\mathbf y,\mathbf z,\ldots)
  =
  A_{v(\mu)}^{ket}
  \bigl((\mathbf r,\mathbf y,\mathbf z,\ldots)+\mathbf e_\mu\bigr).
  \label{eq:core-replica-gluing}
\end{equation}
Here \(\mu=0,\ldots,q-1\), \(\mathbf e_0=\mathbf 0\), and \(\mathbf e_\mu\) for \(\mu\ne0\) is the unit shift in the coordinate associated with physical qubit \(\mu\), with coordinates understood modulo \(n\). In the examples below, unit shifts within the \(\mathbf r\)-, \(\mathbf y\)-, and \(\mathbf z\)-coordinate groups are written as \(\mathbf e_r\), \(\mathbf e_y\), and \(\mathbf e_z\), respectively. Thus the qubits within one GHZ block share an assignment field in the ket and one in the bra, but each physical qubit still supplies its own gluing condition.

We call the GHZ block containing the reference party the \emph{root block}. Since \(A_0\) belongs to block \(v=1\) and \(\mathbf e_0=\mathbf 0\), Eq.~\eqref{eq:core-replica-gluing} gives \(A_1^{ket}=A_1^{bra}\) at every replica site. We call this common assignment field the \emph{root field}. Choosing a different reference party changes the coordinate description of the assignment fields, but not the contraction \(\mathcal Z_n(G)\).

We denote the set of all valid assignments by $\mathcal A_n(G,\mathbf a)$. The gluing conditions are homogeneous linear equations over $\F_2$, so this set is a finite-dimensional vector space over $\F_2$.

Recall that \(D_G(x)=\sum_{(u,v)\in E}x_u x_v\) is the quadratic phase polynomial of the core-graph state in Eq.~\eqref{eq:core-state}, with one term for each edge of \(G\). For a valid assignment \((\mathbf A^{ket},\mathbf A^{bra})\), multiplying the signs of all ket and bra amplitudes gives \((-1)^{Q_{G,n}(\mathbf A^{ket},\mathbf A^{bra})}\), where
\begin{equation}
  Q_{G,n}(\mathbf A^{ket},\mathbf A^{bra})=
  \sum_{\mathbf r,\mathbf y,\mathbf z,\ldots}
  \Bigl[
    D_G\bigl(\mathbf A^{ket}
      (\mathbf r,\mathbf y,\mathbf z,\ldots)\bigr)
    +D_G\bigl(\mathbf A^{bra}
      (\mathbf r,\mathbf y,\mathbf z,\ldots)\bigr)
  \Bigr]
  \in\F_2.
  \label{eq:replica-quadratic-exponent}
\end{equation}
The sum runs over all replica sites $(\mathbf r,\mathbf y,\mathbf z,\ldots)\in\TT^{(q)}$. The amplitudes are real in this representation, so complex conjugation leaves the bra signs unchanged. 
Each ket and bra pair also contributes the common factor $2^{-k}$. Since there are $n^{q-1}$ replica sites, the contraction is
\begin{equation}
\begin{aligned}
  \mathcal Z_n(G)
  &=
  2^{-k n^{q-1}}\Delta_G(n),
  \\
  \Delta_G(n)
  &=
  \sum_{(\mathbf A^{ket},\mathbf A^{bra})
         \in\mathcal A_n(G,\mathbf a)}
  (-1)^{Q_{G,n}(\mathbf A^{ket},\mathbf A^{bra})}.
\end{aligned}
  \label{eq:replica-gauss-representation}
\end{equation} We refer to $\Delta_G(n)$ as {\bf the assignment sum} and remind the reader that the dummy summation variables in it are {\bf assignment fields}; viz. equation~(\ref{defassignmentfields}).

Solving the gluing conditions expresses the values of the assignment fields as linear functions of a set of independent binary variables. Substituting these expressions into $D_G$ leaves a polynomial of degree at most two. Thus the assignment sum $\Delta_G(n)$ is a quadratic Gauss sum over $\F_2$: a sum of signs whose exponents are given by a quadratic polynomial in independent binary variables.

The R\'enyi multi-entropy then follows directly from Eq.~\eqref{eq:entropy-normalization}:
\begin{equation}
    S_n(G)
  =
  \frac{k n^{q-1}-\log_2\Delta_G(n)}
       {(n-1)n^{q-2}}\log 2
       \qquad{\rm with}~
  \Delta_G(n)>0.
\end{equation}
The remaining task is to evaluate the assignment sum $\Delta_G(n)$.

To do so, we first solve the gluing conditions and identify the independent values of the assignment fields. Let $A:S\to\F_2$ be an assignment field whose values on a finite set $S$ can be chosen independently. For a fixed function $\kappa:S\to\F_2$, the basic summation identity is
\begin{equation}
  \sum_{A:S\to\F_2}
  (-1)^{\sum_{s\in S}A(s)\kappa(s)}
  =
  \begin{cases}
    2^{|S|}, & \kappa=0,\\
    0, & \kappa\ne0,
  \end{cases}
  \label{eq:finite-fourier-projection}
\end{equation}
where $\kappa=0$ means $\kappa(s)=0$ for every $s\in S$. If $\kappa=0$, all $2^{|S|}$ assignments contribute $+1$. Otherwise, changing the value of $A$ at any point where $\kappa(s)=1$ reverses the sign, so the terms cancel in pairs.

Suppose the exponent is $\sum_{s\in S}A(s)\kappa(s)+R$, where $\kappa$ and $R$ depend only on the remaining fields. For each fixed choice of those fields, summing over $A$ therefore imposes $\kappa(s)=0$ for every $s\in S$ and gives a factor of $2^{|S|}$ when these conditions hold. The remaining sign is $(-1)^R$. Because the full exponent has degree at most two in the independent binary variables, each $\kappa(s)$ is a linear expression in the remaining variables, possibly including a constant.

If these equations have no solution, the assignment sum vanishes. Otherwise, solve them and continue with the remaining free variables. If the remaining exponent vanishes on all solutions, only the number of solutions is needed: each free binary variable contributes another factor of two.

The same identity also evaluates sums with quadratic terms. After solving the constraints, express the remaining assignment fields in terms of independent binary variables. If the exponent contains a product $\xi\eta$ of two such variables, it can be written as
\begin{equation}
  Q(\xi,\eta,\boldsymbol{\zeta})
  =
  \xi\eta
  +\xi L(\boldsymbol{\zeta})
  +\eta M(\boldsymbol{\zeta})
  +R(\boldsymbol{\zeta}),
  \label{eq:quadratic-pair-decomposition}
\end{equation}
where $\boldsymbol{\zeta}$ collects all the other remaining independent binary variables. Here $L$ and $M$ are $\mathbb{F}_2$-linear expressions, possibly including constants, and $R$ has degree at most two.

For each fixed $\boldsymbol{\zeta}$, summing over $\xi$ imposes
$\eta=L(\boldsymbol{\zeta})$. Consequently,
\begin{equation}
  \sum_{\xi,\eta\in\F_2}
  (-1)^{Q(\xi,\eta,\boldsymbol{\zeta})}
  =
  2(-1)^{
    R(\boldsymbol{\zeta})
    +L(\boldsymbol{\zeta})M(\boldsymbol{\zeta})
  }.
  \label{eq:quadratic-pair-summation}
\end{equation}
This removes two variables and leaves another exponent of degree at most two. Repeating the step removes all products of distinct variables. The remaining linear terms are evaluated by Eq.~\eqref{eq:finite-fourier-projection}, while a constant term $c$ contributes the overall sign $(-1)^c$. Keeping this sign and all factors of two produced during the summations gives the complete assignment sum.

This procedure shows that a quadratic Gauss sum over $\F_2$ is either zero or a signed power of two. Its sign must be determined by a calculation. For the graph families evaluated below, the explicit reductions give $\Delta_G(n)>0$, as required for a real-valued R\'enyi multi-entropy in equation~\eqref{eq:entropy-normalization}. 

The procedure applies to every fixed integer $n$. Obtaining a formula valid for arbitrary $n$ requires understanding how the gluing conditions and the remaining quadratic terms depend on $n$. The two-block and three-block examples below illustrate this calculation before we turn to more general core graphs.
\subsection{Worked examples: Two-block and three-block tree graphs}
\label{subsec:worked-examples}

\paragraph{Two-block graph.}
Consider a core graph consisting of two vertices joined by an edge, with arbitrary positive block sizes $a$ and $b$. The corresponding core-graph state is 
\begin{equation}\label{eq:twoblockstate}
  \ket{\Psi_{\mathrm{edge}}^{a+b}}
  =
  2^{-1}\sum_{x,y\,\in\,\F_2}(-1)^{xy}\ket{x^{\otimes a} y^{\otimes b}}.  
\end{equation}
Using Proposition~\ref{prop:coregraphstate-physicalgraphstate}, it is straightforward to find a physical graph state, which is LC-equivalent to this \textit{core-2} state, as shown in Figure~\ref{fig:core2}.
\begin{figure}[t!]
  \centering
   \includegraphics[width=.6\textwidth]{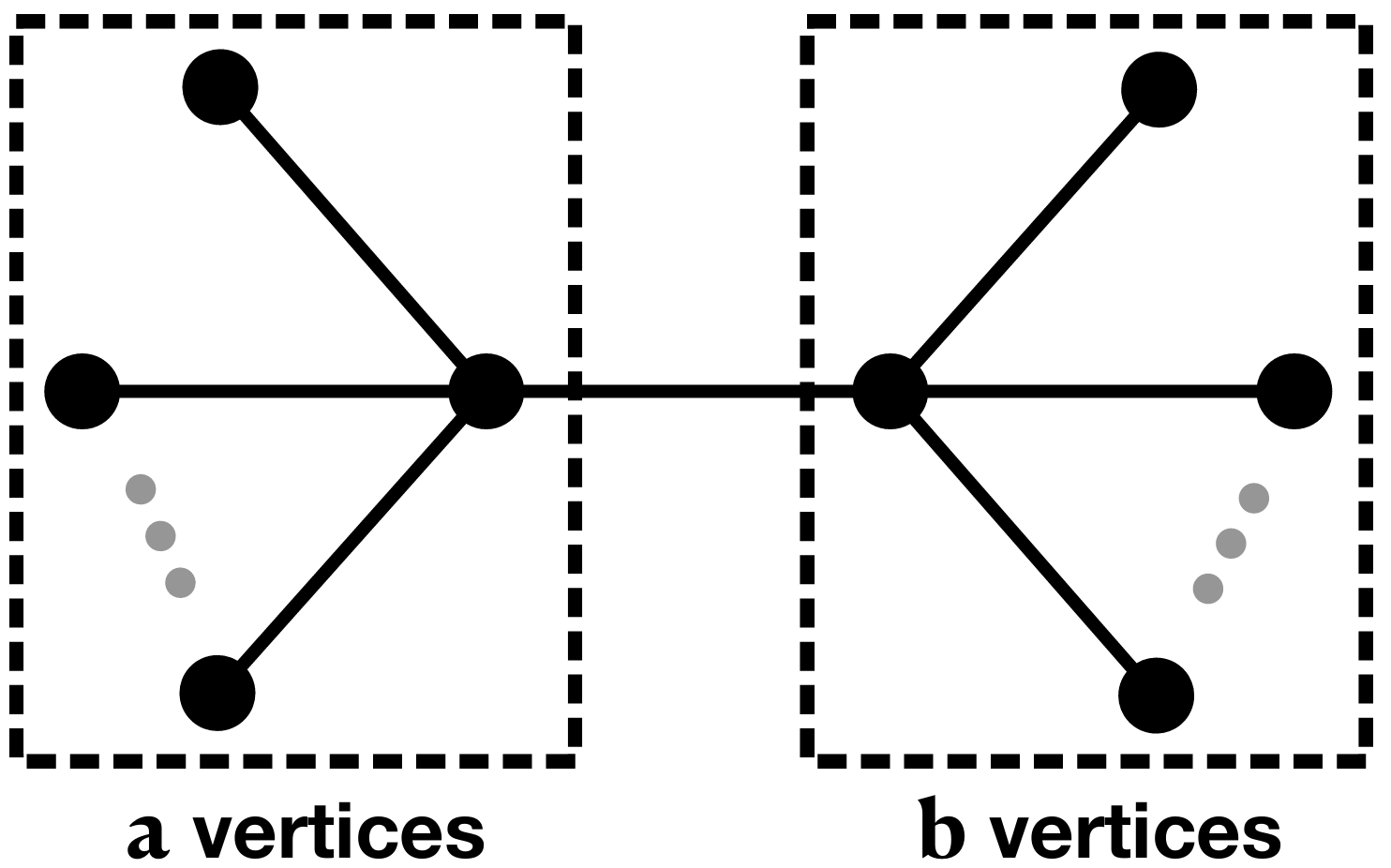}
  \caption{A graph state, which is LC-equivalent to a \textit{core-2} state defined in \eqref{eq:twoblockstate}.}
\label{fig:core2}
\end{figure}
Choose one qubit in the \(x\)-block as reference. The torus coordinates are then
\begin{equation}
    (\mathbf r,\mathbf y)
  =
  (r_1,\ldots,r_{a-1},y_1,\ldots,y_b)\in(\Zn)^{a+b-1},
\end{equation}
corresponding to the two blocks. Valid bit values of these blocks are specified by the assignment fields \(A^{bra},A^{ket},B^{bra},B^{ket}:(\Zn)^{a+b-1}\to\F_2\). Specifically, each torus point $(\mathbf r,\mathbf y)$ is associated with the matrix element $\ket{A^{ket}B^{ket}}\bra{A^{bra}B^{bra}}$. For the bra-ket shift contraction on the toric lattice to be nonzero, the valid assignment fields must satisfy the following gluing conditions in the reference direction and the $\mathbf r,\mathbf y$ directions:
\begin{equation}
\begin{aligned}
        &A^{bra}(\mathbf r,\mathbf y)=A^{ket}(\mathbf r,\mathbf y)=A^{ket}(\mathbf r+\mathbf e_r,\mathbf y)\\
        &B^{bra}(\mathbf r,\mathbf y)=B^{ket}(\mathbf r,\mathbf y+\mathbf e_y),
\end{aligned}
\end{equation}
where $\mathbf e_r$ and $\mathbf e_y$ are arbitrary standard basis vectors in the r- and y- coordinate sectors. Thus, let $s:=|\mathbf y|=y_1+\dots +y_b~(\text{mod } n )$. The coordinate dependence of the bra and ket assignment fields is:
\begin{equation}
\begin{aligned}
        &A^{bra}=A^{ket}=A(\mathbf y), && A: (\Zn)^{b}\to\F_2,\\
        &B^{bra}=B(\mathbf r,s+1),\ B^{ket}=B(\mathbf r,s),\quad && B: (\Zn)^{a}\to\F_2.
\end{aligned}
\end{equation}
This leaves \(2^{n^b}\) independent choices for \(A\) and \(2^{n^a}\) choices for \(B\), so the number of valid assignments is $2^{n^a+n^b}$.

These valid assignments can contribute negatively or positively to the assignment sum, due to the $(-1)$ exponent in \eqref{eq:twoblockstate}. The specific sign is determined by the parity of the total exponent. Define:
\begin{equation}
     N=\sum_{\mathbf r,\mathbf y}
  A(\mathbf y)\Big(B(\mathbf r,s)+B(\mathbf r,s+1)\Big)
  \in\F_2.
\end{equation}
Assignments with \(N=0\) contribute \(+1\), while assignments with \(N=1\) contribute \(-1\).
Thus, the assignment sum can be written as
\begin{equation}
    \Delta_{\mathrm{2}}(n)
  =
  \sum_{\text{valid}}(-1)^N =\sum_{\substack{A:(\Zn)^b\to\F_2\\B:(\Zn)^a\to\F_2}}
  (-1)^{\sum_{\mathbf y}
  A(\mathbf y)\sum_{\mathbf r}\Big(B(\mathbf r,s)+B(\mathbf r,s+1)\Big)}
  \to
  \begin{cases}
  2^{n^b},& \kappa_B=0,\\
  0,& \kappa_B\ne0.
  \end{cases}
\end{equation}
The final equation is based on the Gauss sum property \eqref{eq:finite-fourier-projection}, with
\begin{equation}
    \kappa_{B}(y):=
  \sum_{\mathbf r}\Big(B(\mathbf r,s)+B(\mathbf r,s+1)\Big)
\end{equation}
as a function of $y$.
It then remains to determine the dimension of the solution space of $\kappa_B=0$, since each independent solution contributes $2^{n^b}$ to the assignment sum. Letting $b(s):=\sum_{\mathbf r}\big(B(\mathbf r,s)\big)$, the constraint $\kappa_B=0$ becomes 
\begin{equation}\label{kappaBedge}
    b(s)+b(s+1)=0,
     \qquad \forall\ s\in\mathbb Z_n .
\end{equation}
Before imposing this constraint, $b$ is an arbitrary binary function on $\mathbb Z_n$. Its values $b(s)$ are independent since one may freely vary $B(\mathbf r,s)$ for different values of $s$.
However, \eqref{kappaBedge} demands that
\begin{equation}
    b(s)=b(s+1)
\end{equation}
for every $s$. This renders $b$ a constant function, and compresses its original $n$ degrees of freedom to only one. Therefore, the constraint $\kappa_B=0$ has rank $n-1$. Since the original field $B:(\mathbb Z_n)^{a}\to\mathbb F_2$ has $n^a$ independent binary degrees of freedom, the solution space has dimension $n^a-(n-1)$.

Consequently, combining the Fourier sum $2^{n^b}$ over $A$ and the $n^a-(n-1)$ free binary degrees of freedom of $B$, the assignment sum is:
\begin{equation}
\label{2blockfinal}
    \Delta_{\mathrm{2}}(n)
    =
    2^{n^b}\,2^{\,n^a-(n-1)}
    =
    2^{\,n^a+n^b-(n-1)}.
\end{equation}


\leanverifiedbox{The formula for the assignment sum \(\Delta_2(n)\) has been formally verified in Lean for all \(a,b\ge1\) and integers \(n\ge2\). In the \texttt{Core2} project of our public GitHub
repository~\cite{StabilizerRenyiLean2026}, see
\nolinkurl{Core2/Core2/CombinatorialAssignments.lean} and the declarations \path{card_validAssignments}, \path{card_constraint_kernel_user}, and \path{FormalMathLab.CombinatorialAssignments.signed_imbalance_user}.  These declarations begin from the finite field assignment model; Section 2.3 supplies the manuscript's contraction-to-assignment bridge.}

\paragraph{Three-block tree graph.}
Next consider a three-vertex path. Let the central block have size \(a\), and let the two end blocks have sizes \(b\) and \(c\). Choose a qubit in the central block as the reference. The corresponding state is:
\begin{equation}\label{eq:threeblock}
    \ket{\Psi_{\textrm{3-block}}^{a+b+c}}
  =
  2^{-\frac{3}{2}}\sum_{x,y,z\,\in\,\F_2}(-1)^{xy+xz}\ket{ y^{\otimes b} x^{\otimes a}z^{\otimes c}}.  
\end{equation}
Thus, on the ($a+b+c-1$)-dimensional toric lattice, the vector \(\mathbf r\in(\Zn)^{a-1}\) records the non-reference coordinates in the root block, while \(\mathbf y\in(\Zn)^b\) and \(\mathbf z\in(\Zn)^c\) depict the coordinates in the other two blocks. The assignment fields for the three blocks are again represented with $A,B,C$, with each torus point associated with $\ket{B^{ket} A^{ket} C^{ket}}\bra{B^{bra}A^{bra} C^{bra}}$.

Similar to the two-block example, the gluing conditions first simplify the coordinate dependence of the assignment fields. In the root block, the reference qubit together with the remaining physical qubits identifies the bra and ket root fields across all choices of $\mathbf r$, so the root field becomes independent of the root-block replica coordinates. In each non-root block, the repeated physical qubits impose the same shifted bra-ket identification along every replica direction of that block, so the corresponding assignment field depends on those coordinates only through their cyclic sum. Consequently, let $s=|\mathbf{y}|,t=|\mathbf{z}|$,
\begin{equation}\label{ABCform}
    \begin{aligned}
        &A^{bra}=A^{ket}=A(\mathbf y,\mathbf z), \qquad A: (\Zn)^{b+c}\to\F_2,\\
        &B^{bra}=B(\mathbf r,s+1,\mathbf z),\ B^{ket}=B(\mathbf r,s,\mathbf z),\qquad B: (\Zn)^{a+c}\to\F_2,\\
        &C^{bra}=C(\mathbf r,\mathbf y,t+1),\ C^{ket}=C(\mathbf r,\mathbf y,t),\qquad C: (\Zn)^{a+b}\to\F_2,
\end{aligned}
\end{equation}
resulting in $2^{n^{a+b}+n^{b+c}+n^{a+c}}$ valid assignments in total.

For each choice of assignment, the sign is determined by $\sum_{\mathbf r,\mathbf y, \mathbf z} A(\mathbf y,\mathbf z)\big(B(\mathbf r,s,\mathbf z)+B(\mathbf r,s+1,\mathbf z)+C(\mathbf r,\mathbf y,t)+C(\mathbf r,\mathbf y,t+1)\big)$. The assignment sum is obtained by summing these signed contributions:
      \begin{equation}
      \label{eq:bigeq}
\begin{aligned}
   & \Delta_{\mathrm{3}}(n)
  =
  \sum_{\text{valid}}(-1)^N =\sum_{\substack{A:(\Zn)^{b+c}\to\F_2\\B:(\Zn)^{a+c}\to\F_2\\C:(\Zn)^{a+b}\to\F_2}}
  (-1)^{\sum_{\mathbf y}
  A(\mathbf y,\mathbf z) \kappa_{BC}(\mathbf y,\mathbf z; s, t) 
  }
  =
  \begin{cases}
  2^{n^{b+c}}& {\rm if}~\kappa_{BC}=0\\
  0& {\rm if}~\kappa_{BC}\ne0
  \end{cases} 
  \\&  
  \kappa_{BC}:=\sum_{\mathbf r}\Big(B(\mathbf r,s,\mathbf z)+B(\mathbf r,s+1,\mathbf z)+C(\mathbf r,\mathbf y,t)+C(\mathbf r,\mathbf y,t+1)\Big).
\end{aligned}
\end{equation}
\noindent 
Therefore, similar to the two-block case, the calculation of the assignment sum reduces to determining the dimension of the solution space of $\kappa_{BC}=0$. Again, this is obtained by introducing the summed functions $b(s,\mathbf z)=\sum_{\mathbf r}B(\mathbf r,s,\mathbf z)$ and $c(\mathbf y,t)=\sum_{\mathbf r}C(\mathbf r,\mathbf y,t)$, then rewriting $\kappa_{BC}=0$ as:
\begin{equation}\label{kappabc}
    b(s,\mathbf z)+b(s+1,\mathbf z)+c(\mathbf y,t)+c(\mathbf y,t+1)=0.
\end{equation} 
Section~\ref{subsec:triangle-main} analyzes this space in detail for $a=b=c=2$ to evaluate the triangle's residual quadratic phase. Here we quote the dimension for arbitrary positive block sizes $a,b,c$:
\begin{equation}
    \dim K_n
    =n^{a+c}+n^{a+b}
    -(n-1)(n^b+n^c)
    +(n-1)^2.
\end{equation}
($K_n$ stands for the solution space of $\kappa_{BC}=0$.) Restoring the factor $2^{n^{b+c}}$ from the Fourier sum in (\ref{eq:bigeq}) gives
\begin{equation}
\label{3blocktreefinal}
\Delta_{\mathrm{3}}(n)=2^{X_{\mathrm{3}}(n)}
\quad {\rm with}~~
  X_{\mathrm{3}}(n)
    =
    {n^{b+c}+n^{a+c}+n^{a+b}-(n-1)(n^b+n^c)+(n-1)^2}.
\end{equation}

\leanverifiedbox{
This formula for the assignment sum \(\Delta_3(n)\) has been
formally verified in Lean for all \(a,b,c\ge1\) and integers \(n\ge2\). In the \texttt{Core3} directory of our public GitHub repository~\cite{StabilizerRenyiLean2026} (Lean package \texttt{Core3tree}), the closing declaration is \nolinkurl{No13Migration.signed_imbalance_abc_final_closed} in \nolinkurl{Core3/Core3tree/No13_abc_general.lean}. Eq.~\eqref{eq:replica-gauss-representation} relates the replica contraction to the assignment sum.
}

In summary, we have obtained the assignment sums of two-block graphs and three-block tree graphs in equations~(\ref{2blockfinal}) and (\ref{3blocktreefinal}). The R\'enyi multi-entropies of these states follow from \eqref{eq:entropy-normalization}.

\subsection{Local complementation as a bit-string Fourier identity}\label{subsec:2.5}

Stabilizer states related by a local Clifford transformation have identical replica multi-invariants for the corresponding physical-party partition.  Their graph state representatives, however, need not be described by the same graph: at the graph level, locally Clifford-equivalent graph states are connected by a sequence of local complementations \cite{vandennest2004graph}.  Given a graph \(G\)\footnote{We assume $a_v=1$ for all $v\in V$ in this subsection, i.e. $G=G_{\rm phys}$.} and a vertex \(v\), the local complement \(G*v\) toggles every edge between two neighbors of \(v\) and leaves all other adjacencies unchanged. R\'enyi multi-entropy, as a series of multi-invariant, is invariant under local unitary transformation on the state to be evaluated. As a consistency check, in this subsection, we show explicitly that after translation into the finite field Gauss sum language, this graph operation is implemented by a one-bit Fourier identity.  Consequently, graph states related by local complementation give exactly the same closed Gauss sum, and hence the same multi-invariant as expected.

To see this, we examine how local complementation at $v$ work. Suppose the neighbors of \(v\) carry bits \(x_1,\ldots,x_r\) in computational basis, and set
\begin{equation}
    m=x_1+\cdots+x_r
\end{equation}
as an ordinary integer.  
Let
\begin{equation}
     R=e^{-\pi iX/4}
  =
  \frac1{\sqrt2}
  \begin{pmatrix}
  1&-i\\
  -i&1
  \end{pmatrix},
  \qquad
  R_{y,x}=2^{-1/2}(-i)^{x+y}(-1)^{xy}.
\end{equation}
Then, for every integer \(m\) and \(y\in\{0,1\}\),
\begin{equation}
  (-i)^m\sum_{x=0}^1R_{y,x}(-1)^{xm}
  =
  e^{-\pi i/4}(-1)^{ym+\binom m2},
  \label{eq:lc-one-bit}
\end{equation}
which is easily checked directly for all eight cases: for $y=0,1$ and $0\leq m \leq 3$ (and therefore for all $m$). The term \(xm\) is the old star phase from \(v\) to its neighbors, \(ym\) is the transformed star phase, and \(\binom m2\) is the clique phase toggling the neighbor-neighbor edges. For instance, with three neighbors one has \(m=x_1+x_2+x_3\), and
\begin{equation}
    \binom m2\equiv x_1x_2+x_1x_3+x_2x_3\pmod2,
\end{equation}
which is exactly the phase contributed by toggling the triangle among the three neighbors.

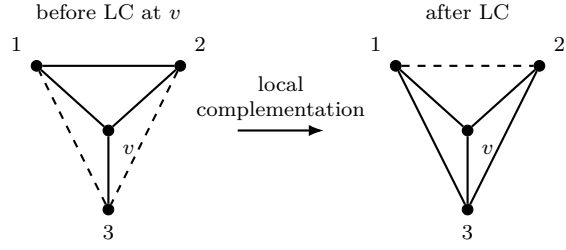
\begin{figure}[h]
\centering
\begin{tikzpicture}[scale=0.95]
\node[smSmallLabel] at (0,1.65) {\scriptsize before LC at \(v\)};
\node[smOpenVertex,label=below right:{\scriptsize \(v\)}] (v) at (0,0) {};
\node[smOpenVertex,label=above left:{\scriptsize 1}] (x1) at (-1,0.9) {};
\node[smOpenVertex,label=above right:{\scriptsize 2}] (x2) at (1,0.9) {};
\node[smOpenVertex,label=below:{\scriptsize 3}] (x3) at (0,-1.1) {};
\draw[smEdge] (v)--(x1);
\draw[smEdge] (v)--(x2);
\draw[smEdge] (v)--(x3);
\draw[smEdge] (x1)--(x2);
\draw[smEdge,dashed] (x2)--(x3);
\draw[smEdge,dashed] (x3)--(x1);

\draw[smArrow] (1.8,0) -- (3.0,0);
\node[smSmallLabel, align=center] at (2.4,0.45) {local\\complementation};

\begin{scope}[shift={(5.0,0)}]
\node[smSmallLabel] at (0,1.65) {\scriptsize after LC};
\node[smOpenVertex,label=below right:{\scriptsize \(v\)}] (V) at (0,0) {};
\node[smOpenVertex,label=above left:{\scriptsize 1}] (y1) at (-1,0.9) {};
\node[smOpenVertex,label=above right:{\scriptsize 2}] (y2) at (1,0.9) {};
\node[smOpenVertex,label=below:{\scriptsize 3}] (y3) at (0,-1.1) {};
\draw[smEdge] (V)--(y1);
\draw[smEdge] (V)--(y2);
\draw[smEdge] (V)--(y3);
\draw[smEdge,dashed] (y1)--(y2);
\draw[smEdge] (y2)--(y3);
\draw[smEdge] (y3)--(y1);
\end{scope}
\end{tikzpicture}
\caption{Local complementation at \(v\): edges among the neighbors of \(v\) are toggled.  In the coefficient sum this graph move is implemented by the one-bit Fourier identity in Eq.~\eqref{eq:lc-one-bit}.}
\label{fig:local-complementation}
\end{figure}
\FloatBarrier
We shall make a comparison between this identity and the finite Fourier projection trick used in Section \ref{subsec:finite-fourier-gauss}. While the earlier trick that sums up all the real characters \((-1)^{A\kappa}\) acted as a filter (imposing a mathematical constraint), this new trick acts as a transformer. It physically redraws the graph by reversing the connectivity of the neighbors of the vertex $v$. In both cases, our assignment calculus is a dictionary between stabilizer operations and finite sums over bit-strings.

To verify that the assignment sum is invariant under local complementation, we evaluate a modified sum in two different orders. Define \(\widetilde\Delta_{G,v}(n)\) by inserting the local Clifford kernel above on every ket coefficient and its conjugate on every bra coefficient in the closed replica contraction. If we evaluate the contraction links first, unitarity cancels the inserted transformations with their adjoints, leaving the original assignment sum \(\Delta_G(n)\). If instead we evaluate the coefficient sums first, the Fourier identity~\eqref{eq:lc-one-bit} transforms the graph state coefficients into those of the local complement \(G*v\), yielding \(\Delta_{G*v}(n)\). Since both orders evaluate the same modified sum, we obtain
\begin{equation}
    \Delta_G(n)=\widetilde\Delta_{G,v}(n)=\Delta_{G*v}(n).
\end{equation}
Thus the bit-string assignment formulation preserves invariance under local complementation, in agreement with the local-unitary invariance of the underlying replica contraction.

\section{Tree graphs}
\label{sec:trees}

We now evaluate the R\'enyi multi-entropy for tree core graphs, with each physical qubit treated as a separate party. Section \ref{subsec:subtreeformula} gives a formula valid for every integer \(n\geq2\)
and proves it by induction on the number of core vertices. Section~\ref{Sec32Properties} examines the \(n\to1\) limit and shows that tree graph states in distinct local Clifford orbits can have identical R\'enyi multi-entropy spectra.

\subsection{The connected subtree formula}
\label{subsec:subtreeformula}

For any given core tree graph, the exponent of its assignment sum is expressed as a sum over its connected subtrees. The contribution of each
subtree is determined by its internal edges and by the block sizes of the vertices outside it. The precise statement is as follows:

\begin{theorem}[tree graph formula]
\label{thm:tree-main}
Let \(T\) be a finite connected tree serving as the core graph of a stabilizer graph state. Each vertex \(v\in V(T)\) has block size \(a_v\ge1\), and we set \(q=\sum_{v\in V(T)}a_v\ge2\). Let \(\Delta_T(n)\) be the assignment sum for the \(q\)-party R\'enyi multi-entropy contraction, with each physical qubit treated as a separate party. For every integer \(n\ge2\),
\begin{equation}
  \Delta_T(n)=2^{X_T(n)},
\end{equation}
where
\begin{equation}
\boxed{
  X_T(n)=
  \sum_{\substack{\varnothing\ne W\subseteq V(T)\\
                  T[W]\ \mathrm{connected}}}
  (1-n)^{|W|-1}
  n^{\sum_{v\notin W}a_v}.}
  \label{eq:tree-exponent}
\end{equation}
Here \(T[W]\) is the induced subgraph of \(T\) on the vertex subset \(W\), and \(|W|\) is the number of vertices in \(W\).
\end{theorem}

\leanverifiedbox{
The tree graph formula has been formally verified in Lean 4 for arbitrary finite trees, all positive block sizes, and every integer \(n\ge2\). This verification starts from the finite-field assignment formulation, whose derivation from the graph-state replica contraction is given in Section~\ref{sec:calculus}. For historical reasons, the formal verification follows a different proof from the one presented below. We will formalize the present proof in a future version. The corresponding Lean theorem is \nolinkurl{TreeGraph.ShearedFinal.tree_imbalance_rawExpandedEdgeSignedSum_final_closed} in \nolinkurl{TreeGraph/TreeGraph/ShearedFinal.lean} in our public GitHub repository~\cite{StabilizerRenyiLean2026}. The file path is relative to the repository root; build instructions are provided in \nolinkurl{BUILDING.txt}.}

We prove Theorem \ref{thm:tree-main} by induction on the number \(k\) of core vertices. To do so, we first study the effect of eliminating a leaf from the tree. This allows us to derive a recursion relation that expresses the assignment sum exponent \(X_T\) for a \(k\)-vertex core tree in terms of the corresponding exponents for core trees with \(k-1\) vertices. We then establish \eqref{eq:tree-exponent} by verifying that it satisfies the same recursion relation and produces the correct initial values of GHZ states and two block graph states.


\subsection*{The elimination of a leaf}

We start by explicitly illustrating the elimination of a leaf. Fix an arbitrary core vertex $r\in V(T)$ as the root of $T$, and let it supply the reference qubit.

Denote the leaf vertex as $l$ and its neighboring vertex as $p$\footnote{For generality, we consider the case where $p$ is not the root. The special case $p=r$ will be discussed later in this subsection. }. $\mathbf{x}_l\in \mathbb{Z}_n^{a_l}$ and $\mathbf{x}_p\in \mathbb{Z}_n^{a_p}$ are the replica coordinates on the torus with dimensions $a_l$ and $a_p$, which correspond to their respective GHZ blocks. To describe the bit value of these blocks in the bras and kets, we introduce the assignment fields $A^{ket}_p,A^{ket}_l,A^{bra}_p,A^{bra}_l: \mathbb{Z}_n^{q-1}\rightarrow\mathbb{F}_2$ on each point of the torus. The torus coordinates are organized into $(\bm{\rho},\mathbf{x}_p,\mathbf{x}_l)$, where $\bm{\rho}\in\mathbb{Z}_n^{q-a_l-a_p-1}$ collects the replica coordinates associated with all other blocks in the tree.

We then identify the gluing conditions of $A^{ket}_p,A^{ket}_l,A^{bra}_p,A^{bra}_l$. By the bra-ket contractions in the directions of $\mathbf{x}_p,\mathbf{x}_l$, these fields satisfy:
\begin{equation}
    \begin{aligned}
        A_p^{bra}(\bm{\rho},\mathbf{x}_p,\mathbf{x}_l)
        &=
        A_p^{ket}(\bm{\rho},\mathbf{x}_p+\mathbf e_p,\mathbf{x}_l)
        \\
        A_l^{bra}(\bm{\rho},\mathbf{x}_p,\mathbf{x}_l)
        &=
        A_l^{ket}(\bm{\rho},\mathbf{x}_p,\mathbf{x}_l+\mathbf e_l).
    \end{aligned}
\end{equation}
Thus, $A_p^{bra}$ and $A_p^{ket}$ depend only on $\bm{\rho},\mathbf{x}_l$, and $s:=|\mathbf{x}_p|\in\mathbb{Z}_n$. They can be rewritten as
\begin{equation}
    A_p^{bra}=A_p(\bm{\rho},s+1,\mathbf{x}_l),
    \qquad
    A_p^{ket}=A_p(\bm{\rho},s,\mathbf{x}_l).
\end{equation}
Similarly, one obtains
\begin{equation}
    A_l^{bra}=A_l(\bm{\rho},\mathbf{x}_p,t+1),
    \qquad
    A_l^{ket}=A_l(\bm{\rho},\mathbf{x}_p,t),
\end{equation}
where $t:=|\mathbf{x}_l|\in\mathbb{Z}_n$.

We then calculate how summing the leaf field $A_l$ contributes to the assignment sum $\Delta_T$. The total parity exponent of $T$ is $N=N_0+\sum_{\bm{\rho},\mathbf{x}_p,\mathbf{x}_l} \big(A_p^{bra}A_l^{bra}+A_p^{ket}A_l^{ket}\big)$, where $N_0$ does not contain the field $A_l$. Thus:
\begin{equation}\label{longDelta}
    \begin{aligned}
        \Delta_T
        &=\sum_{\mathrm{valid}}(-1)^N
        \\
        &=\sum_{\{A_v\}_{v\ne l}}^{\mathrm{valid}}
        (-1)^{N_0}
        \sum_{A_l}
        (-1)^{
        \sum_{\bm{\rho},\mathbf{x}_p,\mathbf{x}_l}
        \Big(
        A_p(\bm{\rho},s+1,\mathbf{x}_l)
        A_l(\bm{\rho},\mathbf{x}_p,t+1)
        +
        A_p(\bm{\rho},s,\mathbf{x}_l)
        A_l(\bm{\rho},\mathbf{x}_p,t)
        \Big)
        }
        \\
        &=\sum_{\{A_v\}_{v\ne l}}^{\mathrm{valid}}
        (-1)^{N_0}
        \sum_{A_l}
        (-1)^{
        \sum_{\bm{\rho},\mathbf{x}_p,t}
        A_l(\bm{\rho},\mathbf{x}_p,t)
        \sum_{\mathbf{x}_l:\Sigma\mathbf{x}_l=t}
        \Big(
        A_p(\bm{\rho},s+1,\mathbf{x}_l-\mathbf{e}_l)
        +
        A_p(\bm{\rho},s,\mathbf{x}_l)
        \Big)
        }.
    \end{aligned}
\end{equation}
The final equation comes from the fact that the sum over $(\bm{\rho},\mathbf{x}_p,\mathbf{x}_l)$ allows us to translate the first term in the exponent along $\mathbf{e}_l$ such that $A_l$ can be extracted. 

From here, we use the Gauss-sum property \eqref{eq:finite-fourier-projection} proved in Section \ref{subsec:finite-fourier-gauss}. Let $S$ be the finite index set parametrized by coordinates $(\bm{\rho},\mathbf{x}_p,t)$, whose free dimension is $|S|=n^{q-a_l}$. Let
\begin{equation}
    \kappa_p(\bm{\rho},s,t)
    :=
    \sum_{\mathbf{x}_l:\Sigma\mathbf{x}_l=t}
    \Big(
    A_p(\bm{\rho},s,\mathbf{x}_l)
    +
    A_p(\bm{\rho},s+1,\mathbf{x}_l-\mathbf{e}_l)
    \Big).
\end{equation}
Then:
\begin{equation}\label{eq:treeGauss}
    \sum_{A_l}
    (-1)^{
    \sum_{\bm{\rho},\mathbf{x}_p,t}
    A_l(\bm{\rho},\mathbf{x}_p,t)
    \kappa_p(\bm{\rho},s,t)
    }
    =
    \begin{cases}
        2^{n^{q-a_l}},&
        \kappa_p(\bm{\rho},s,t)=0
        \quad\text{for all }(\bm{\rho},s,t),
        \\
        0,&\text{otherwise}.
    \end{cases}
\end{equation}
Introducing the summed function $B_p(\bm{\rho},s,t)= \sum_{\mathbf{x}_l:\Sigma\mathbf{x}_l=t}A_p(\bm{\rho},s,\mathbf{x}_l)$, $\kappa_p=0$ becomes:
\begin{equation}\label{eq:Bconstraint}
    B_p(\bm{\rho},s,t)
    +
    B_p(\bm{\rho},s+1,t-1)=0.
\end{equation}
This constraint has a similar structure to the examples in Section \ref{subsec:worked-examples}. Here the cyclic shift $\tau$ is on the combined direction of $s$ and $t$ together, with $\tau(a_p)(\bm{\rho},s,t)=B_p(\bm{\rho},s+1,t-1)$, and $\kappa_p=(1+\tau)B_p(\bm{\rho},s,t)$. A point located on the torus again goes through a length-$n$ orbit under this shift. For fixed $(\bm{\rho},s+t)$, $\kappa_p=0$ requires that $B_p$ be constant on this orbit, and consequently imposes $n-1$ independent constraints on the field $A_p$ during the summation of the leaf field.

If \(p=r\), \textit{i.e.}, the leaf to be eliminated connects directly to the root vertex, the root block has \(\mathbf x_p\in(\mathbb Z_n)^{a_p-1}\) while \(\bm\rho\in(\mathbb Z_n)^{q-a_p-a_l}\). Due to bra-ket contractions on the reference and $\mathbf x_p$ directions, the gluing condition for the root field $A_p$ now imposes
\begin{equation}
  A_p^{bra}=A_p^{ket}=A_p(\bm\rho,\mathbf x_l),
  \label{eq:tree-parent-root-field}
\end{equation}
where $A_p$ is independent of $\mathbf x_p$. and the same calculation gives
\begin{equation}
  B_p(\bm\rho,t)
    :=\sum_{\mathbf x_l:\,|\mathbf x_l|=t}
       A_p(\bm\rho,\mathbf x_l),
  \qquad
  \kappa_p(\bm\rho,t)
    =B_p(\bm\rho,t)+B_p(\bm\rho,t-1).
  \label{eq:tree-root-leaf-constraint}
\end{equation}
In this case, eliminating the leaf still contributes a $2^{n^{q-a_l}}$ factor, and imposes the constraint $\kappa_p(\bm\rho,t)=0$, meaning that $B_p(\bm\rho,t)$ is constant as a function of $t$.

This leaf elimination can be repeated: once all descendants of an internal vertex $v$ have been eliminated, $v$ becomes a leaf of the remaining tree, and can be eliminated using the same method\footnote{Note that $A_v$ is still a field on the replica torus: the replica coordinates corresponding to the already-eliminated leaves do not vanish, but are instead reorganized into $\bm\rho$ as spectator directions.}. The only difference lies in \eqref{longDelta} to \eqref{eq:treeGauss}: instead of freely summed over all binary fields on the torus, $A_v$ is now summed over the linear subspace $L_v$ satisfying the constraints inherited from its descendants. This gives a factor $2^{\dim_{\mathbb F_2}L_v}$, and leaves a linear constraint on the neighboring field.

To briefly summarize, we have established the leaf elimination rule of general tree graphs: summing a leaf field produces a power-of-two factor and leaves a linear constraint on the neighboring field whose nontrivial cyclic part has rank \(n-1\). This elimination process can be continued for internal vertices, until only the root field remains. From here, all the $(-1)$ exponents have been summed out, and summation over the root field contributes a power-of-two factor with the size of its constrained subspace.

\subsubsection*{A recursion relation obeyed by the tree graph assignment sums}
The leaf elimination property motivates us to derive a recursion relation that allows the assignment sum exponent of a tree to be expressed by those of smaller trees.

Consider the three trees shown in Figure~\ref{fig:tree-single-leaf-recursion}. In the larger tree $T$, choose a core vertex which has a single leaf descendant as the reference-supplying root. These vertices are labeled as $r$ and $l$, with block sizes $a_r$ and $a_l$ respectively. The smaller tree \(U\) is obtained by deleting \(l\) from $T$. \(U^+\) denotes the same core tree as $U$, only with the block size at \(r\) increased from \(a_r\) to \(a_r+1\). 
\begin{figure}[htbp]
  \centering
  \begin{tikzpicture}[
    x=1cm,y=1cm,
    sl vertex/.style={smOpenVertex,draw=black,fill=black},
    sl edge/.style={smEdge},
    sl label/.style={font=\small,inner sep=2pt}
  ]
    \foreach \offset/\blocksize/\treename [count=\panel from 0] in {
      0/{a_r}/{T},
      4.1/{a_r}/{U},
      8.2/{a_r+1}/{U^+}
    } {
      \begin{scope}[xshift=\offset cm]
        \node[font=\normalsize] at (0,2.15) {$\treename$};
        \node[sl vertex] (r) at (0,1.35) {};
        \node[sl label,anchor=south] at (0,1.48)
          {$r\;(\blocksize)$};
        \node[sl vertex] (p) at (-0.30,0.30) {};
        \node[sl vertex] (a) at (-1.00,-0.65) {};
        \node[sl vertex] (b) at (0.40,-0.65) {};
        \node[sl vertex] (c) at (-1.40,-1.60) {};
        \node[sl vertex] (d) at (-0.60,-1.60) {};
        \node[sl vertex] (e) at (0.10,-1.60) {};
        \node[sl vertex] (f) at (0.70,-1.60) {};
        \foreach \u/\v in {p/r,p/a,p/b,a/c,a/d,b/e,b/f}
          \draw[sl edge] (\u)--(\v);
        \ifnum\panel=0
          \node[sl vertex] (l) at (1.05,0.45) {};
          \node[sl label,anchor=west] at (1.18,0.45) {$l\;(a_l)$};
          \draw[sl edge] (r)--(l);
        \fi
      \end{scope}
    }
  \end{tikzpicture}
  \caption{Illustration of the three core tree graphs in the recursion relation. Deleting the leaf $l$ from $T$ gives $U$, and further increasing the block size at $r$ by one gives $U^+$. The remaining components of the three graphs are identical.}
  \label{fig:tree-single-leaf-recursion}
\end{figure}
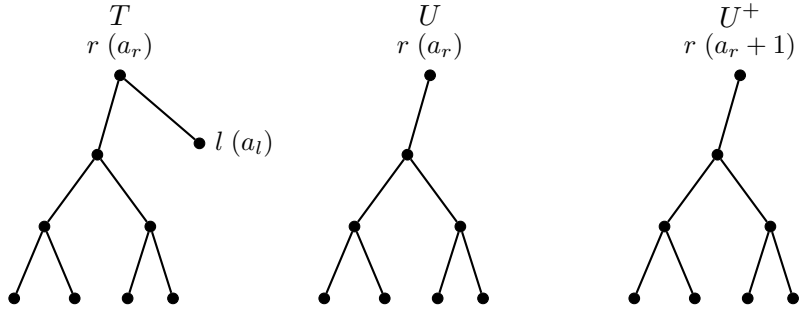

We first evaluate the assignment sum for $U$. Specifically, consider the stage where all descendants of $r$ have been eliminated, but the root field $A_r$ is left unsummed. According to the leaf elimination rule, these summed fields each contribute a power-of two factor, and the fields neighboring to $r$ impose constraints on $A_r$: 
\begin{equation}
  M_U(A_r):=2^P\mathbf1_{A_r\in L_r},
  \qquad d=\dim_{\mathbb F_2}L_r,
  \qquad
  \Delta_U=\sum_{A_r}M_U(A_r)=2^{P+d}
  \label{eq:tree-fixed-root-sum}
\end{equation}
Here $P$ denotes the collective power-of-two contributions of all non-root fields, and \(L_r\) is the constrained linear space of the root field.

We then analyze how adding a root-block coordinate from $U$ to $U^+$ affects the assignment sum \eqref{eq:tree-fixed-root-sum}. For every non-root field, the added coordinate belongs to $\bm \rho$ as an additional spectator direction. During the elimination process previously described, whether this summed field corresponds to an original leaf or an internal vertex, this additional direction only expands the constrained linear space it lives in, but cannot serve as a shift direction in the constraints. As a consequence, this added coordinate gives \(n\) independent copies of the non-root field sum. The root field, on the other hand, does not rely on the root-block coordinates, as explained by the gluing condition \eqref{eq:tree-parent-root-field}. $L_r$ is thus unaffected by this added coordinate. Consequently, 
\begin{equation}\label{DeltaU+}
  \Delta_{U^+}
    =\sum_{A_r}M_U(A_r)^n=2^{nP+d}.
\end{equation}

It remains to study the assignment sum for $T$. This time, the additional leaf $l$ adds $a_l$ directions to the torus, which are packed into the coordinates $\mathbf{x}_l$. First, similar to the case of $U^+$, these $a_l$ directions serve as spectators for the original non-root fields, and do not change the original shift directions in the constraints of these fields. Thus, they first contribute $n^{a_l}$ copies of the non-root field sum $P$. Meanwhile, as demonstrated at the end of the leaf elimination argument, eliminating $l$ contributes another $2^{n^{q-a_l}}$ factor, and imposes a constraint on the root field $A_r$. As shown in \eqref{eq:tree-root-leaf-constraint}, this requires the summed root field $B_r$ to be constant on $|\mathbf{x}_l|$. To count the remaining choices, select one vector \(\mathbf x_l\) with \(|\mathbf x_l|=t\) for each \(t\in\Zn\). The root fields at the other \(n^{a_l}-n\) choices of $\mathbf x_l$ and the common value of \(B_r\) as functions of non-root, non-$\mathbf x_l$ coordinates can be chosen freely in the original constrained space \(L_r\), and the fields at the selected vectors are then uniquely determined. This leaves \((n^{a_l}-n+1)d\) binary degrees of freedom for the root field. Combining these contributions gives:
\begin{equation}
    X_T=n^{q-a_l}+n^{a_l}P+(n^{a_l}-n+1)d,\qquad \Delta_T=2^{X_T}.
  \label{eq:tree-X-T}
\end{equation}

Rewrite $\Delta_U$ and $\Delta_{U^+}$ into their exponents $X_U$, $X_{U^+}$. By eliminating $P$ and $d$, we obtain the recursive relation:
\begin{equation}
    X_T(n)
    =
    n^{q-a_l}+(n^{a_l}-n)X_U(n)+X_{U^+}(n).
  \label{eq:X-recursion}
\end{equation}
This formula governs the relationship between the assignment sum of a larger tree and those of its subtrees with fewer vertices.

\subsubsection*{Base case}
For an inductive argument on the number of core vertices of a tree, the base case is a single GHZ core. Its assignment field is a single free binary variable, giving \(X=1\) and \(\Delta=2\) for every positive block size. This agrees with formula~\eqref{eq:tree-exponent}, where the core vertex itself is the only connected subgraph in the summation.

\subsubsection*{Inductive step}

We prove that formula~\eqref{eq:tree-exponent} obeys the same recursion relationship as the assignment sum, that is equation~(\ref{eq:X-recursion}). This is the inductive step for the validity of~\eqref{eq:tree-exponent}. The recursion relationship that governs~\eqref{eq:tree-exponent} is read off by a direct calculation:

\begin{equation}
  \begin{aligned}
    \tilde{X}_T(n)
    &\equiv
    n^{q-a_l}+n^{a_l}X_U(n)
    +\sum_{\substack{
        W\subseteq V(T),\,l\in W\\
        T[W]\ \mathrm{connected},\,|W|\ge2
      }}
      (1-n)^{|W|-1}
      n^{\sum_{v\in V(T)\setminus W}a_v}
    \\[3pt]
    &=
    n^{q-a_l}
    +n^{a_l}\sum_{\substack{
        \varnothing\ne W\subseteq V(U)\\
        U[W]\ \mathrm{connected},\,r\notin W
      }}
      (1-n)^{|W|-1}
      n^{\sum_{v\in V(U)\setminus W}a_v}
    \\
    &\quad
    +(n^{a_l}+1-n)\sum_{\substack{
        W\subseteq V(U),\,r\in W\\
        U[W]\ \mathrm{connected}
      }}
      (1-n)^{|W|-1}
      n^{\sum_{v\in V(U)\setminus W}a_v}
    \\[3pt]
    &=
    n^{q-a_l}+(n^{a_l}-n)X_U(n)
    +\sum_{\substack{
        \varnothing\ne W\subseteq V(U)\\
        U[W]\ \mathrm{connected},\,p\notin W
      }}
      (1-n)^{|W|-1}
      n^{1+\sum_{v\in V(U)\setminus W}a_v}
    \\
    &\quad
    +\sum_{\substack{
        W\subseteq V(U),\,r\in W\\
        U[W]\ \mathrm{connected}
      }}
      (1-n)^{|W|-1}
      n^{\sum_{v\in V(U)\setminus W}a_v}
    \\[3pt]
    &=
    n^{q-a_l}+(n^{a_l}-n)\tilde{X}_U(n)+\tilde{X}_{U^+}(n).
  \end{aligned}
\end{equation}

\subsection{Properties of R\'enyi multi-entropy on trees}
\label{Sec32Properties}
We now examine what information about a tree graph is retained by its R\'enyi multi-entropy spectrum. We briefly examine the \(n\to1\) limit and show that states in distinct local Clifford orbits can have identical multi-entropy spectra for every integer \(n\ge2\).

Note that an arbitrary tree graph state always admits a representation where its core graph is identical to the original tree graph. In this representation, each vertex forms a single block, so that \(a_v=1\) for all vertices. Therefore, let $T$ be the original tree graph, set \(\mathcal N=|V(T)|\) and denote the number of size-$s$ connected subtrees of $T$ with $c_s(T)$. Theorem \ref{thm:tree-main} can then be rewritten as:
\begin{equation}
\label{treefinalresult}
  X_T(n)=
  \sum_{s=1}^{\mathcal N}c_s(T)(1-n)^{s-1}n^{\mathcal N-s}.
\end{equation}

Thus, the information captured by the R\'enyi multientropies boils down to counting connected subtrees of all sizes. Observe that $X_T(1)=c_1(T)=\mathcal N$ universally. Beyond that, this information is not sufficient to distinguish locally Clifford-inequivalent tree graph states \cite{Wagner2014,Eisenstat2006,Martin2008}. One example, which involves 10 qubits, is given in Figure~\ref{fig:tree-counterexample}. The two tree graphs displayed there are not locally Clifford-equivalent yet they have the same numbers of $s$-sized connected subtrees:
\begin{equation}
(c_1,\ldots,c_{10})=(10,9,10,12,14,14,12,8,4,1)
\end{equation}
This is the smallest example, in the sense that all locally Clifford-inequivalent trees with 9 or fewer vertices have distinct numbers of connected subtrees.

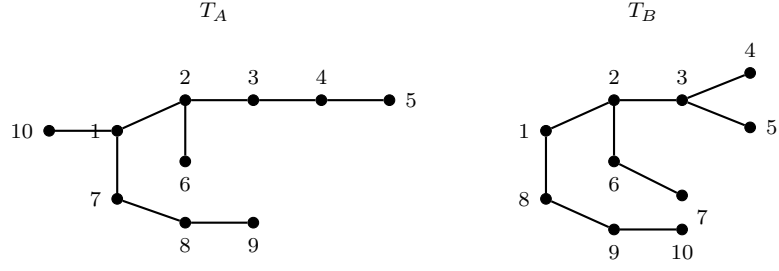
\begin{figure}[!h]
\centering
\begin{tikzpicture}[scale=0.9]
\node[smSmallLabel] at (1.4,1.75) {\(\,T_A\)};

\node[smOpenVertex,label=left:{\scriptsize 1}] (A1) at (0,0) {};
\node[smOpenVertex,label=above:{\scriptsize 2}] (A2) at (1.0,0.45) {};
\node[smOpenVertex,label=above:{\scriptsize 3}] (A3) at (2.0,0.45) {};
\node[smOpenVertex,label=above:{\scriptsize 4}] (A4) at (3.0,0.45) {};
\node[smOpenVertex,label=right:{\scriptsize 5}] (A5) at (4.0,0.45) {};
\node[smOpenVertex,label=below:{\scriptsize 6}] (A6) at (1.0,-0.45) {};
\node[smOpenVertex,label=left:{\scriptsize 7}] (A7) at (0,-1.0) {};
\node[smOpenVertex,label=below:{\scriptsize 8}] (A8) at (1.0,-1.35) {};
\node[smOpenVertex,label=below:{\scriptsize 9}] (A9) at (2.0,-1.35) {};
\node[smOpenVertex,label=left:{\scriptsize 10}] (A10) at (-1.0,0.0) {};

\foreach \u/\v in {A1/A2,A1/A7,A1/A10,A2/A3,A2/A6,A3/A4,A4/A5,A7/A8,A8/A9}
  \draw[smEdge] (\u)--(\v);

\begin{scope}[shift={(6.3,0)}]
\node[smSmallLabel] at (1.4,1.75) {\(\,T_B\)};

\node[smOpenVertex,label=left:{\scriptsize 1}] (B1) at (0,0) {};
\node[smOpenVertex,label=above:{\scriptsize 2}] (B2) at (1.0,0.45) {};
\node[smOpenVertex,label=above:{\scriptsize 3}] (B3) at (2.0,0.45) {};
\node[smOpenVertex,label=above:{\scriptsize 4}] (B4) at (3.0,0.85) {};
\node[smOpenVertex,label=right:{\scriptsize 5}] (B5) at (3.0,0.05) {};
\node[smOpenVertex,label=below:{\scriptsize 6}] (B6) at (1.0,-0.45) {};
\node[smOpenVertex,label=below right:{\scriptsize 7}] (B7) at (2.0,-0.95) {};
\node[smOpenVertex,label=left:{\scriptsize 8}] (B8) at (0,-1.0) {};
\node[smOpenVertex,label=below:{\scriptsize 9}] (B9) at (1.0,-1.45) {};
\node[smOpenVertex,label=below:{\scriptsize 10}] (B10) at (2.0,-1.45) {};

\foreach \u/\v in {B1/B2,B1/B8,B2/B3,B2/B6,B3/B4,B3/B5,B6/B7,B8/B9,B9/B10}
  \draw[smEdge] (\u)--(\v);
\end{scope}
\end{tikzpicture}
\caption{Two non-isomorphic, non-LC-equivalent tree graph states with the same number of arbitrary-sized connected subtrees.  By the tree-graph formula, the complete R\'enyi multi-entropy spectrum cannot distinguish them.}
\label{fig:tree-counterexample}
\end{figure}

\section{Loop graphs}
\label{sec:loops}
We now extend our analysis from tree graphs to graphs with loops. Section~\ref{subsec:universal-known-graph-formula} discusses how, in the case of odd replica index $n$, results for tree graphs can be leveraged for calculating R\'enyi multi-entropies of loop graphs. In Section~\ref{subsec:triangle-main}, we use the triangle loop as a case study and explain how and why the multi-R\'enyis of loop graphs exhibit different behavior for odd and even $n$. Section~\ref{subsec:core4-main} extends that analysis to four- and five-vertex loop cores and, in the case of odd $n$, for more general loop cores including arbitrary odd cycles. Finally, Section~\ref{subsec:c6-global} turns to the six-cycle graph and discusses the challenge of continuing the calculation to that case and beyond.

\subsection{Decomposing graphs into loop cores}
\label{subsec:universal-known-graph-formula}

In equation~\eqref{eq:tree-exponent}, adding contributions from induced connected  subgraphs\footnote{Given a graph $H$ and $U$ a subset of its vertex set, the induced subgraph \(H[U]\) contains the vertices in \(U\) (and no others) and all edges of \(H\), which join them.} turned out to be a productive strategy for computing the R\'enyi multi-entropies of tree graphs. In this section, we apply the same strategy for graphs with loops. 

For a connected core graph \(H\), let \(X_H(n)\) be the exponent of its assignment sum: $\Delta_H(n) = 2^{X_H(n)}$. Let \(\Gamma_H(n)\) be the exponent stripped off the contributions from the induced connected proper subgraphs:
\begin{equation}
\boxed{\displaystyle
~\Gamma_H(n)
=
X_H(n)
-
\sum_{\substack{
    \varnothing\neq U\subsetneq V(H)\\
    H[U]\ \mathrm{connected}
}}
\Gamma_{H[U]}(n)\,
n^{\sum_{v\in V(H)\setminus U}a_v}.~
}
\label{eq:cluster-extraction}
\end{equation}
Observe that the appearance of $\Gamma_{H[U]}(n)$ on the right hand side makes this definition recursive in the number of vertices of $H$. 
Both \(X_H(n)\) and \(\Gamma_H(n)\) are implicitly functions of the GHZ block sizes $a_v$. Tautologically, the exponent for a graph \(G\) is:
\begin{equation}
X_G(n)
=
\sum_{\substack{
    \varnothing\neq U\subseteq V(G)\\
    G[U]\ \mathrm{connected}
}}
\Gamma_{G[U]}(n)\,
n^{\sum_{v\in V(G)\setminus U}a_v}.
\label{eq:loop-support-expansion}
\end{equation}
The powers of \(n\), which multiply each coefficient account for the GHZ blocks outside the chosen subgraph.

For trees, all graphs with the same number of vertices have the same coefficient \(\Gamma_T\), independently of their shape and GHZ-block sizes. We call this \emph{\(\Gamma\)-degeneracy}. Different trees have different exponents \(X_T\), but the differences only arise from counting connected induced proper subgraphs.

With loops present, \(\Gamma_H\) itself develops a dependence on the connectivity of the graph. Moreover, as we shall see in the next subsection, odd and even \(n\) exhibit qualitatively different behaviors, with the odd branch significantly simpler than the even branch.

\paragraph{A formula for odd $n$} In the remainder of this subsection, we concentrate on the case where $n$ is odd. Let \(H\) be any connected core graph. An edge is called a \emph{bridge} if deleting that edge (while keeping all vertices) splits \(H\) into two disconnected parts. We write \(\mathcal{B}(H)\) for the set of bridges and 
\begin{equation}
    b(H)=|\mathcal{B}(H)|
\end{equation}
for their number.

Remove all bridges from \(H\) and call the remaining connected components \(L_1,\ldots L_k\). By construction, each \(L_i\) is a loop subgraph with no bridges. The $L_i$ may contain more than one loop, which may or may not share edges. (In the latter case, two loops may be ``tangent'' at a vertex.) The $L_i$ may also be isolated vertices.

Let \(\Gamma_{L_i}\) be the individual coefficient for the connected bridge-free subgraph $L_i$, as defined in equation~(\ref{eq:cluster-extraction}). For isolated vertices, we have $\Gamma_{\{\textrm{single vertex}\}}=1$, consistently with (\ref{eq:tree-exponent}). Assuming odd $n \geq 3$,  we have:
%
%
\begin{equation}
\boxed{\displaystyle
~\Gamma_H(n)
=
(1-n)^{b(H)}\,
\prod_{i=1}^{k}\Gamma_{L_i}(n)~
}
\label{eq:odd-bridge-factorization}
\end{equation}
We offer this formula without proof. Documentation proving its validity is available upon request.

A tree is a special case of this expression. Because every edge of a tree \(T\) is a bridge, the bridge-free subgraphs are isolated vertices. Using $b(T) = |V(T)|-1$ we therefore find $\Gamma_T(n) = (1-n)^{|V(T)|-1}$, which reproduces result~(\ref{treefinalresult}). 
As we saw in Section~\ref{sec:trees}, in trees this expression is also valid for even \(n\). 

Our next task is to build a dictionary of loop cores $\Gamma_L$. We also wish to return to the case of even replica index $n$, which we have thus far set aside.

%

\subsection{The triangle core: A residual quadratic Gauss sum}
\label{subsec:triangle-main}
We begin with the triangle, which is the simplest loop core. 
We show how summing over the linearly appearing fields leaves a residual quadratic Gauss sum and how this sum can be evaluated. The techniques developed here provide a foundation for analyzing the four-vertex, five-vertex, and more general loop cores considered in Section \ref{subsec:core4-main}.

To make the calculation explicit, we first consider the simplest case where each vertex of the triangle core graph corresponds to a GHZ block of size $2$, as shown in Figure \ref{fig:triangle-222}. We remark that this $6$-qubit state is the minimal example of a graph state, which is not LC-equivalent to a tree. (If a vertex $v$ in a triangular core graph has $a_v - 1 = 0$ leaves, applying local complementation at $v$ transforms the graph to a tree.) 

Denote the three core vertices by $v_1,v_2,v_3$. Let $v_1$ be the root vertex whose block supplies the reference qubit. The remaining replica coordinates are then:
\begin{equation}
    r\in\mathbb{Z}_n,\quad \mathbf{y}=(y_1,y_2)\in \mathbb{Z}_n^2,\ \text{and}\quad  \mathbf{z}=(z_1,z_2)\in \mathbb{Z}_n^2.
\end{equation}
Here $r$ is the non-reference coordinate in the root block while $\mathbf{y}$ and $\mathbf{z}$ correspond to the $2$-qubit blocks of $v_2$ and $v_3$. The assignment fields $A(r,\mathbf{y},\mathbf{z})$, $B(r,\mathbf{y},\mathbf{z})$ and $C(r,\mathbf{y},\mathbf{z})$ record the bit choices for the three GHZ blocks in the bras and kets associated with the torus sites. The bra-ket contractions impose the following gluing conditions:
\begin{equation}
    \begin{aligned}
        A^{bra}(r,\mathbf{y},\mathbf{z})&=A^{ket}(r,\mathbf{y},\mathbf{z})\\
    A^{bra}(r,\mathbf{y},\mathbf{z})&=A^{ket}(r+1,\mathbf{y},\mathbf{z})\\
        B^{bra}(r,\mathbf{y},\mathbf{z})&=B^{ket}(r,\mathbf{y}+\mathbf{e}_i,\mathbf{z})\\
        C^{bra}(r,\mathbf{y},\mathbf{z})&=C^{ket}(r,\mathbf{y},\mathbf{z}+\mathbf{e}_j) \qquad {\rm for}~~i,j\in\{1,2\}.
    \end{aligned}
\end{equation}
Here $\mathbf{e}_i$ and $\mathbf{e}_j$ are the standard unit vectors in the $y_i$ and $z_j$ directions. Denoting $s:=|\mathbf{y}|$ and $t:=|\mathbf{z}|$, these assignment fields simplify to:
\begin{equation}
    \begin{aligned}
        A^{bra}=A^{ket}=A(\mathbf{y},\mathbf{z}),&\\
        B^{bra}=B(r,s,\mathbf{z}),&\ ~~B^{ket}=B(r,s+1,\mathbf{z}),\\
        C^{bra}=C(r,\mathbf{y},t),&\ ~~C^{ket}=C(r,\mathbf{y},t+1).
    \end{aligned}
\end{equation}
The contribution of these fields to the $(-1)$ phase is determined by the three triangle edges. The exponent can be decomposed into two parts:
\begin{equation}
\begin{aligned}
    N=N_0+Q_{BC}&=\sum_{r,\mathbf{y},\mathbf{z}}A(\mathbf{y},\mathbf{z})\big(B(r,s,\mathbf{z})+B(r,s+1,\mathbf{z})+C(r,\mathbf{y},t)+C(r,\mathbf{y},t+1)\big)\\
    &+\sum_{r,\mathbf{y},\mathbf{z}}\big(B(r,s,\mathbf{z})C(r,\mathbf{y},t)+B(r,s+1,\mathbf{z})C(r,\mathbf{y},t+1)\big).
\end{aligned}
\end{equation}
The sum in the upper line, denoted $N_0$, is linear in the root field $A(\mathbf{y},\mathbf{z})$. As explained in previous tree graph examples, the summation over $A$ consequently imposes linear constraints on $B$ and $C$. Meanwhile, $Q_{BC}$ is a genuinely new, residual quadratic coupling between the non-root fields $B$ and $C$. Performing the summation over $A$, we obtain:
\begin{equation}
\label{sumqbc}
\begin{aligned}
    \Delta_{C_3}^{2,2,2}&=\sum_{\text{valid}}(-1)^N\\
    &=\sum_{B,C}(-1)^{Q_{BC}}\sum_A(-1)^{\sum_{r,\mathbf{y},\mathbf{z}}A(\mathbf{y},\mathbf{z})\big(B(r,s,\mathbf{z})+B(r,s+1,\mathbf{z})+C(r,\mathbf{y},t)+C(r,\mathbf{y},t+1)\big)}\\
    &=2^{n^4} \sum_{\substack{B,C\\\kappa_{BC}=0}}(-1)^{Q_{BC}}.
\end{aligned}
\end{equation}
The quantity, which is set to 0 in the final summation, reads:
\begin{equation}\label{kappa}
    \kappa_{BC}(\mathbf y,\mathbf z):=\sum_{r}\Big(B(r,s,\mathbf z) +B(r,s+1,\mathbf z)+C(r,\mathbf y,t)+C(r,\mathbf y,t+1)\Big).
\end{equation}

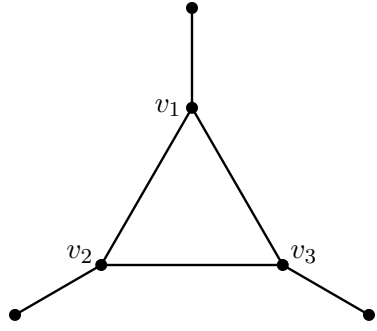
\begin{figure}[t]
\centering
\begin{tikzpicture}[
    scale=1.2,
    line width=0.9pt,
    vertex/.style={circle, fill=black, inner sep=1.6pt},
    every label/.style={font=\small}
]

\coordinate (v1) at (0,{sqrt(3)});
\coordinate (v2) at (-1,0);
\coordinate (v3) at (1,0);

\coordinate (u1) at (0,{sqrt(3)+1.1});
\coordinate (u2) at (-1.95,-0.55);
\coordinate (u3) at (1.95,-0.55);

\draw (v1) -- (v2) -- (v3) -- cycle;

\draw (v1) -- (u1);
\draw (v2) -- (u2);
\draw (v3) -- (u3);

\node[vertex] at (v1) {};
\node[vertex] at (v2) {};
\node[vertex] at (v3) {};

\node[xshift=-9pt] at (v1) {$v_1$};
\node[xshift=-8pt,yshift=4pt] at (v2) {$v_2$};
\node[xshift=8pt,yshift=4pt] at (v3) {$v_3$};

\node[vertex] at (u1) {};
\node[vertex] at (u2) {};
\node[vertex] at (u3) {};

\end{tikzpicture}
\caption{The $6$-qubit loop graph whose core graph is a triangle with block sizes $a_1 = a_2 = a_3 = 2$. It is the minimal example of a loop core graph, which is not LC-equivalent to a tree. Its assignment sum is calculated in Section \ref{subsec:triangle-main}.}
\label{fig:triangle-222}
\end{figure}

We compute sum~(\ref{sumqbc}) in two steps:
\begin{enumerate}
    \item[(1)] Determine the dimension of the solution space of $\kappa_{BC}=0$, temporarily disregarding the phase $(-1)^{Q_{BC}}$. 
    \item[(2)] Observe that $\sum_{p,q\,\in\,\mathbb{F}_2}(-1)^{pq} = 2$, which is $1/2$ of $\sum_{p,q\,\in\,\mathbb{F}_2} (1) = 4$. Therefore, each isolated bilinear pair (such as $pq$ in the schematic example) takes away one power of 2 from sum~(\ref{sumqbc}). We count the independent {\bf hyperbolic pairs} that occur in $Q_{BC}$. 
\end{enumerate}
In the end, the assignment sum is $2^{n^4 + {\rm step(1)} - {\rm step(2)}}$.




\subsubsection*{Step $(1)$: Dimension of solution space}
Let $K_n$ denote the solution space of the constraint $\kappa_{BC}=0$. Before evaluating the dimension of $K_n$, we simplify this constraint by introducing the $r$-summed functions:
\begin{equation}
    \overline{B}(s,\mathbf{z}):=\sum_rB(r,s,\mathbf{z}),\qquad \overline{C}(\mathbf{y},t):=\sum_rC(r,\mathbf{y},t).
\end{equation}
Then, let $\tau_s,\tau_t$ be the cyclic shift operators along the $s$ and $t$ directions, such that
\begin{equation}\label{cyclicshiftst}
   \tau_s(\overline{B})(s,\mathbf{z})=\overline{B}(s+1,\mathbf{z}),\qquad \tau_t(\overline{C})(\mathbf{y},t)=\overline{C}(\mathbf{y},t+1).
\end{equation}
For notational convenience, we further define $\mathcal{D}_s:=1-\tau_s$, and $\mathcal{D}_t:=1-\tau_t$. Thus, the constraint $\kappa_{BC}=0$ can be written as:\footnote{Note that the fields $B$ and $C$, as well as $\overline{B}$ and $\overline{C}$ are all binary functions whose values lie in $\mathbb{F}_2$. Therefore the addition and subtraction are identical in $\kappa_{BC}$, and also later in the simplification of $Q_{BC}$.}
\begin{equation}\label{vanishingkapparetell}
    \mathcal{D}_s \overline{B}(s,\mathbf{z})+\mathcal{D}_t \overline{C}(\mathbf{y},t)=0.
\end{equation}

The degrees of freedom in $K_n$ under this constraint can be divided into two sectors. The first sector contains choices of ($B,C$) that disappear during the summation over $r$ from $B,C$ to $\overline{B},\overline{C}$, and are thus {\bf invisible} to the constraint \eqref{vanishingkapparetell}. The second sector consists of the solutions that are {\bf visible} in $\overline{B},\overline{C}$, and non-trivially satisfies \eqref{vanishingkapparetell}. 

We begin with the \textit{invisible sector}, which is determined by $\overline{B}=\overline{C}=0$. The summation over $r$ can be conveniently described by a linear map
\begin{equation}
    \sigma: \mathbb{F}_2^{\mathbb{Z}_n}\rightarrow \mathbb{F}_2\quad \text{with}\quad  \sigma(f)=\sum_{r\,\in\,\mathbb{Z}_n}f(r)
\end{equation}
for any $f\in\mathbb{F}_2^{\mathbb{Z}_n}$. Under each fixed choice of $(s,\mathbf{z})$, the function ${B}(\cdot,s,\mathbf{z})$ can be viewed as an element of $\mathbb{F}_2^{\mathbb{Z}_n}$, and $\overline{B}(s,\mathbf{z})=\sigma \big(B(\cdot,s,\mathbf{z})\big)$. Thus, the invisible sector is precisely characterized by the kernel of $\sigma$, which has dimension $n-1$. As we vary $s,\mathbf{z}$, there are $n^3$ independent choices, so the invisible sector of $B$ has $n^3(n-1)$ free dimensions in total. The same argument applies to $C$, which contributes another $n^3(n-1)$ dimensions to the invisible sector.

We then move on to analyze the \textit{visible sector}. From \eqref{vanishingkapparetell}, it can be easily verified that each pair of solutions $(\overline{B}, \overline{C})$ takes the following general form:
\begin{equation}\label{visibleform}
    \overline{B}=B_0(\mathbf{z})+\mathcal{D}_tH(s,t), \qquad \overline{C}=C_0(\mathbf{y})+\mathcal{D}_sH(s,t),
\end{equation}
where $B_0,C_0$ are arbitrary functions. Note that when evaluating the dimension of this sector, for a given $H(s,t)$, any variation of the form $\Tilde{H}(s,t)=H(s,t)+h_s(s)+h_t(t)$ should be regarded as equivalent to $H(s,t)$. This is because after applying $\mathcal{D}_t,\mathcal{D}_s$, the degrees of freedom in the additional terms $h_s(s)$, $h_t(t)$ are already contained in the arbitrary choices of $B_0$ and $C_0$. Thus, these terms can be absorbed into $B_0,C_0$. Summarizing the above contributions to the visible sector, the free choices of $B_0(\mathbf{z})$ and $C_0(\mathbf{y})$ each contribute $n^2$ dimensions. The independent degrees of freedom in $H(s,t)$ are
\begin{equation}
    n^2-(2n-1)=(n-1)^2.
\end{equation}
Here $2 \times n$ counts the redundant dimensions carried by each of $h_s(s)$ and $h_t(t)$ whereas the $-1$ corrects for double-counting the constant function in both of them.

We combine the contributions from the invisible and visible sectors and find:
\begin{equation}
    \dim K_n=2n^3(n-1)+2n^2+(n-1)^2=2n^4-2n^3+3n^2-2n+1.
    \label{dimkn}
\end{equation}

\subsubsection*{Step $(2)$: Enumerating hyperbolic pairs}
%

To identify the independent hyperbolic pairs, we first simplify $Q_{BC}$. Define: 
\begin{equation}\label{bcrst}
    b(r,s,t):=\sum_{\mathbf{z}:|\mathbf{z}|=t}B(r,s,\mathbf{z}) ,\qquad c(r,s,t):=\sum_{\mathbf{y}:|\mathbf{y}|=s}C(r,\mathbf{y},t).
\end{equation}
Then $Q_{BC}$ only depends on the summed functions $b$ and $c$: 
\begin{equation}\label{QBClong}
\begin{aligned}
    &Q_{BC}=\sum_{r,\mathbf{y},\mathbf{z}}\big(B(r,s,\mathbf{z})C(r,\mathbf{y},t)+B(r,s+1,\mathbf{z})C(r,\mathbf{y},t+1)\big)\\
    &=\sum_{r,s,t}\sum_{\mathbf{z}:|\mathbf{z}|=t}B(r,s,\mathbf{z})\sum_{\mathbf{y}:|\mathbf{y}|=s}C(r,\mathbf{y},t)+\sum_{r,s,t}\sum_{\mathbf{z}:|\mathbf{z}|=t}B(r,s+1,\mathbf{z})\sum_{\mathbf{y}:|\mathbf{y}|=s}C(r,\mathbf{y},t+1)\\
    &=\sum_{r,s,t}b(r,s,t)c(r,s,t)+\sum_{r,s,t}b(r,s+1,t)c(r,s,t+1)\\
    &=\sum_{r,s,t}b(r,s,t)\big(c(r,s,t)+c(r,s-1,t+1)\big).
\end{aligned}
\end{equation}
The last equation follows because translating $s$ does not affect the summation over $r,s,t$ in the second term. 

Eecall the cyclic shift operator $\tau_s,\tau_t$ defined in \eqref{cyclicshiftst}. Define:
\begin{equation}
  \langle b,c\rangle
  :=\sum_{r,s,t} b(r,s,t)c(r,s,t)\in\mathbb F_2.  
\end{equation}
The residual phase can then be written as:
\begin{equation}
    Q_{BC}=\Big\langle b,(I_r\otimes \mathcal{D}_{st})c\Big\rangle, \quad \text{where}\quad \mathcal{D}_{st}:=1-\tau_s^{-1}\tau_t
\end{equation}
and $I_r$ is the unit operator on $r$. Thus \(Q_{BC}\) depends on \(B,C\) only through \(b,c\). Variations that leave \(b,c\) unchanged are already included in \(\dim K_n\), but do not affect the phase. To count the hyperbolic pairs, we therefore study the allowed values of \(b,c\) under the constraint \(\kappa_{B,C}=0\).

In fact, the same visible-invisible decomposition for $B,C$ above can now be applied to the summed functions $b$ and $c$. For each fixed $(s,t)$, $b(\cdot, s,t)$ and $c(\cdot, s,t)$ are again elements of $\mathbb{F}_2^{\mathbb{Z}_n}$ and their invisible sectors lie in $\ker \sigma$. To describe the visible sectors, denote
\begin{equation}
\overline{b}(s,t):=\sum_rb(r,s,t),\qquad \overline{c}(s,t):=\sum_rc(r,s,t).    
\end{equation}
Then, directly from \eqref{visibleform}:
\begin{equation}\label{bcbar}
\overline{b}=\begin{cases}b_0(t)+\mathcal{D}_tH(s,t),&n\text{ odd},\\b_0(t),&n\text{ even},\end{cases}\qquad
\overline{c}=\begin{cases}c_0(s)+\mathcal{D}_sH(s,t),&n\text{ odd},\\c_0(s),&n\text{ even}.\end{cases}
\end{equation}
Here $b_0$ and $c_0$ are arbitrary functions. 

The relevant decomposition of $\mathbb{F}_2^{\mathbb{Z}_n}$ depends on the parity of $n$.

\paragraph{Odd $n$.} For $n$ is odd, let $\mathbf{1}:=(1,\dots,1)\,\in\,\mathbb{F}_2^{\mathbb{Z}_n}$. Since 
\begin{equation}
\sigma(\mathbf{1})=\sum_{r\in \mathbb{Z}_n}\mathbf{1}(r)=1~~(\text{mod}\ 2),  
\end{equation}
we have $\mathbf{1}\notin\ker\sigma$. Moreover, by definition, for every element $f\in \ker\sigma$, $\langle\mathbf{1},f\rangle=\sum_rf(r)=0$.
Thus, $\mathbf{1}$ is orthogonal to $\ker \sigma$. Since $\dim\ker\sigma=n-1$, its orthogonal complement is one-dimensional. The full space $\mathbb{F}_2^{\mathbb{Z}_n}$  therefore admits the orthogonal direct-sum decomposition
\begin{equation}\label{odddecompose}
    \mathbb{F}_2^{\mathbb{Z}_n}=\ker\sigma\perp\text{span}\{\mathbf{1}\}.
\end{equation}
Accordingly, $b(r,s,t)$ and $c(r,s,t)$ take the following form:
\begin{equation}
    b(r,s,t)= b_{\text{inv}}(r,s,t)+\mathbf{1}(r) \overline{b}(s,t),\quad c(r,s,t)= c_{\text{inv}}(r,s,t)+\mathbf{1}(r) \overline{c}(s,t),
\end{equation}
where $b_{\text{inv}},c_{\text{inv}}\in\ker\sigma$ are the invisible components under the summation over $r$. Thus,
\begin{equation}
    Q_{BC}=\Big\langle b,(I_r\otimes \mathcal{D}_{st})c\Big\rangle=\Big\langle b_{\text{inv}},(I_{\ker \sigma}\otimes \mathcal{D}_{st})c_{\text{inv}}\Big\rangle+\langle\mathbf{1},\mathbf{1}\rangle\Big\langle\overline{b}, \mathcal{D}_{st}\overline{c}\Big\rangle. 
\end{equation}
Here the two inner products in the second term are taken over the $r$ coordinate and over the $s,t$ coordinates, respectively. Note that the mixed couplings between the invisible and the visible sectors vanish due to orthogonality \eqref{odddecompose}. Recall that $\langle\mathbf{1},\mathbf{1}\rangle=1$.
Substituting \eqref{bcbar} in $\langle\overline{b}, \mathcal{D}_{st}\overline{c}\rangle$ produces four terms, each of which individually vanishes. This can be shown by summing a cyclic difference over \(s\) or \(t\).

Therefore, $Q_{BC}=\langle b_{\text{inv}},(I_{\ker \sigma}\otimes \mathcal{D}_{st})c_{\text{inv}}\rangle$. Hence the visible sectors generate no hyperbolic pairs, and all reductions arise from the coupling between the invisible sectors of $b$ and $c$. By Gaussian elimination, the number of hyperbolic pairs generated by this remaining coupling, denoted $h_{odd}$, is given by the rank of $I_{\ker \sigma}\otimes \mathcal{D}_{st}$. Since each direction of $r$ in $\ker \sigma$ carries an independent copy of $\mathcal{D}_{st}$, 
\begin{equation}
\label{hodd}
    h_{odd}=\rank (I_{\ker \sigma}\otimes \mathcal{D}_{st})=\dim ({\ker \sigma})\rank ( \mathcal{D}_{st})=(n-1)n(n-1)=n(n-1)^2.
\end{equation}
We used $\rank (\mathcal{D}_{st})=n(n-1)$. This is because the shift $\tau_s^{-1}\tau_t$ divides the $(s,t)$ torus into $n$ disjoint orbits of length $n$, each characterized by a fixed value of $s+t$. On each orbit, $\mathcal{D}_{st}=1-\tau_s^{-1}\tau_t$ has a one-dimensional constant kernel, which gives the rank $n-1$.

\paragraph{Even $n$.} Now the decomposition of $\mathbb{F}_2^{\mathbb{Z}_n}$ is different because
\begin{equation}
\sigma(\mathbf{1})=\sum_{r\in \mathbb{Z}_n}\mathbf{1}(r)=0~~(\text{mod}\ 2) 
\end{equation}
so $\mathbf{1}\in \ker\sigma$. We choose another element $e\in \mathbb{F}_2^{\mathbb{Z}_n}$ satisfying $\sigma(e)=1$, and hence $\langle\mathbf{1},e\rangle=1$. In this way, we obtain the two-dimensional nondegenerate space: $\text{span}\{\mathbf{1},e\}$. Denote its orthogonal complement in $\mathbb{F}_2^{\mathbb{Z}_n}$ by $E:=\text{span}\{\mathbf{1},e\}^\perp$, with $\dim E=n-2$. This gives a new orthogonal decomposition:
\begin{equation}\label{evendecompose}
    \mathbb{F}_2^{\mathbb{Z}_n}=E\perp \text{span}\{\mathbf{1},e\}.
\end{equation}
Under this decomposition, $b(r,s,t)$ and $c(r,s,t)$ can be written as
\begin{align}
    b(r,s,t) & = b_1(r,s,t)+\mathbf{1}(r) u_b(s,t)+e(r) \overline{b}(s,t) 
    \nonumber \\ 
    c(r,s,t) & = c_1(r,s,t)+\mathbf{1}(r) u_c(s,t)+e(r) \overline{c}(s,t),
\end{align}
where $b_1,c_1$ take values in $E$ while $u_b,u_c$ are arbitrary functions of $(s,t)$. Here $e(r)$ pairs precisely with the visible $\overline{b}$ and $\overline{c}$, since $e$ is the only nonvanishing direction after summation over $r$. For even $n$, \eqref{bcbar} reduces to 
\begin{equation}
    \overline{b}=b_0(t),\qquad \overline{c}=c_0(s).
\end{equation}
Thus, $Q_{BC}$ becomes
\begin{equation}
\begin{aligned}
    &Q_{BC}=\Big\langle b,(I_r\otimes \mathcal{D}_{st})c\Big\rangle\\&=\Big\langle b_{1},(I_{E}\otimes \mathcal{D}_{st})c_{1}\Big\rangle+\langle e,\mathbf{1}\rangle\langle u_b,\mathcal{D}_{st}\overline{c}\rangle+\langle \mathbf{1},e\rangle\langle \overline{b},\mathcal{D}_{st}u_c\rangle+\langle e,e\rangle\langle\overline{b}, \mathcal{D}_{st}\overline{c}\rangle\\
    &=\Big\langle b_{1},(I_{E}\otimes \mathcal{D}_{st})c_{1}\Big\rangle+\langle u_b,\mathcal{D}_{st}\overline{c}\rangle+\langle \overline{b},\mathcal{D}_{st}u_c\rangle
\end{aligned}
\end{equation}
Again, there are no coupling terms between $E$ and $\text{span}\{\mathbf{1},e\}$ because of the orthogonality in \eqref{evendecompose}. In the second line, $\langle \mathbf{1},e\rangle=\langle e,\mathbf{1}\rangle=\sigma(e)=1$ while the last term vanishes by cyclic shift cancellation since $\overline{b}$ and $\overline{c}$ depend only on $t$ and $s$, respectively. At last, the three remaining terms involve disjoint sets of independent variables and give rise to independent hyperbolic pairs.

The first term gives $\rank(I_E\otimes \mathcal{D}_{st})$, just as it did for odd $n$. To study the two additional couplings $\langle u_b,\mathcal{D}_{st}\overline{c}\rangle$ and $\langle \overline{b},\mathcal{D}_{st}u_c\rangle$, define:
\begin{equation}
    U_b(s):=\sum_tu_b(s,t), \qquad U_c(t):=\sum_su_c(s,t).
\end{equation}
Because $\overline{c}$ depends only on $s$ while $\overline{b}$ depends only on $t$, these two inner products reduce to:
\begin{equation}
\begin{aligned}
    &\langle u_b,\mathcal{D}_{st}\overline{c}\rangle=\sum_sU_b(s)\Big(\overline{c}(s)-\overline{c}(s-1)\Big)
    =\langle U_b,\,(1-\tau_s^{-1})\,\overline{c}\,\rangle\\
    &\langle \overline{b},\mathcal{D}_{st}u_c\rangle=\sum_t\overline{b}(t)\Big(U_c(t)-U_c(t+1)\Big)
    =\langle \overline{b},\,(1-\tau_t)\,U_c\,\rangle.
\end{aligned}
\end{equation}
By Gaussian elimination, these two couplings respectively contribute $\rank(1-\tau_s^{-1})$ and $\rank(1-\tau_t)$  hyperbolic pairs to the reduction. Together, 
\begin{equation}
\label{heven}
\begin{aligned}
    &h_{even}=\rank(I_E\otimes \mathcal{D}_{st})+\rank(1-\tau_s^{-1})+\rank(1-\tau_t)\\
    &=(n-2)n(n-1)+2(n-1)=(n-1)(n^2-2n+2).
\end{aligned}
\end{equation}

\subsubsection*{Assembling the answer}
The assignment sum is $2^{n^4 + {\rm step(1)} - {\rm step(2)}}$. Using (\ref{dimkn}), (\ref{hodd}) and (\ref{heven}), we find:
\begin{align}
    \log_2 \Delta_{C_3,\text{(odd)}}^{2,2,2} 
    & 
    =3n^4-3n^3+5n^2-3n+1 
    \label{oddnbranch}\\
    \log_2 \Delta_{C_3,\text{(even)}}^{2,2,2}
    & 
    =3n^4-3n^3+6n^2-6n+3
    \label{evennbranch}
\end{align}
This calculation is easily extended to triangle core graphs with arbitrary block sizes $a,b,c$:
\begin{equation}
\begin{aligned}
\log_2\Delta_{C_3}^{a,b,c}(n)
={}&n^{a+b}+n^{a+c}+n^{b+c}\\
&-\begin{cases}
(n-1)(n^a+n^b+n^c-3n+3),
&n\text{ even and }a,b,c\ge2,\\[1mm]
(n-1)(n^a+n^b+n^c-2n+1),
&\text{otherwise}.
\end{cases}
\end{aligned}
\label{abclog}
\end{equation}
The even $n$ case requires $a,b,c\ge2$ because the root block must retain a non-reference direction and each non-root block must provide an additional replica direction needed for the couplings $\langle u_b,\mathcal{D}_{st}\overline{c}\rangle$ and $\langle \overline{b},\mathcal{D}_{st}u_c\rangle$. 

From \eqref{abclog}, we can extract the reusable triangle coefficient $\Gamma_{C_3}$. For the triangle core graph, the connected subgraph exponent decomposes as:
\begin{equation}
\begin{aligned}
    X_{C_3}(a,b,c,n)=\log_2\Delta_{C_3}^{a,b,c}(n)=&n^{a+b}+n^{b+c}+n^{a+c}+(1-n)(n^a+n^b+n^c)\\&+\Gamma_{C_3}(a,b,c,n).
\end{aligned}
\end{equation}
The first three terms correspond to the three singleton subgraphs while the product collects the contributions from the three connected two-vertex subgraphs, each consisting of two vertices joined by an edge. Subtracting off these lower-support contributions from $\log_2\Delta_{C_3}^{a,b,c}(n)$, we obtain:
\begin{equation}
\Gamma_{C_3}(a,b,c,n)=
\begin{cases}
3(n-1)^2,&n\text{ even and }a,b,c\ge2,\\
(n-1)(2n-1),&\text{otherwise}.
\end{cases}
\label{trianglegamma}
\end{equation}
Equation~\eqref{trianglegamma} establishes the first loop entry in the dictionary of individual graph coefficients~(\ref{eq:cluster-extraction}). 

\leanverifiedbox{The assignment sum for a triangle core graph with GHZ blocks of arbitrary sizes \eqref{abclog} has been formally verified in Lean for all \(a,b,c\ge1\) and \(n\ge2\) in the \texttt{Core3loop} project of our public GitHub
repository~\cite{StabilizerRenyiLean2026}. The corresponding public theorem is
\path{Core3loop.DeltaABC_closed_by_cases_main}. The specialization to three GHZ$_2$ blocks is declared as \path{Core3loop.DeltaABC_regression_222_by_cases}, and its statement agrees with the independently formalized result \path{No16.Delta16_closed_by_parity}.  The Lean verification starts from valid assignments satisfying the replica gluing conditions, and carries the calculation through to an exact evaluation of the assignment sum, yielding \eqref{abclog}. The derivation of the gluing conditions from the replica tensor contraction, as well as the subsequent entropy normalization, is given in Section \ref{sec:calculus} and is not part of these Lean declarations.}

\subsection{A dictionary of four-vertex, five-vertex, and more general loop cores}
\label{subsec:core4-main}

We exhibit explicit full-graph coefficients for the four connected four-vertex loop cores and the eighteen connected five-vertex loop cores, together with an odd-\(n\) formula for an arbitrary odd cycle \(C_{2h+1}\). Using equation~\eqref{eq:loop-support-expansion} and the definitions, these results give analytic expressions for the R\'enyi multi-entropies of the corresponding families of states with arbitrary positive GHZ-block sizes. 

We list the four-vertex coefficients in Table~\ref{tab:core4-gamma-dictionary}. Labels \(a,b,c,d\) in the table denote the positive GHZ-block sizes at the corresponding vertices. For a condition \(\mathcal P\), we let \(\mathbf 1_{\mathcal P}=1\) when \(\mathcal P\) holds and \(\mathbf 1_{\mathcal P}=0\) otherwise.
 Coefficients for five-vertex and more general loop cores, including the available even-\(n\) expressions for the five-vertex graphs, are collected in Appendix~\ref{app:core5-dictionary}.

The calculation uses the Gauss-sum techniques developed for the triangle in Section~\ref{subsec:triangle-main}. After evaluating the assignment sum of each core, we subtract the contributions of its proper connected induced subgraphs, as prescribed in Section~\ref{subsec:universal-known-graph-formula}. We present the resulting coefficients without detailed derivations and illustrate their use below.

To illustrate how the dictionary is used, consider the triangle-plus-leaf graph in the second row. Its proper connected induced subgraphs consist of four single vertices, four edges, the two paths \(b\!-\!a\!-\!d\) and \(c\!-\!a\!-\!d\), and the triangle on the vertices labelled \(a,b,c\). The single-vertex, edge, and three-vertex path coefficients are \(1\), \(1-n\), and \((1-n)^2\), respectively. Equation~\eqref{eq:loop-support-expansion} therefore gives:
\begin{equation}
    \begin{aligned}
  X_{\rm TL}(a,b,c,d;n)={}&
  n^{b+c+d}+n^{a+c+d}+n^{a+b+d}+n^{a+b+c}\\
  &+(1-n)\bigl(
    n^{c+d}+n^{b+d}+n^{a+d}+n^{b+c}
  \bigr)\\
  &+(1-n)^2\bigl(n^c+n^b\bigr)
  +\Gamma_{C_3}(a,b,c;n)n^d\\
  &+\Gamma_{\rm TL}(a,b,c,d;n).
\end{aligned}
\end{equation}
Substituting the second row of
Table~\ref{tab:core4-gamma-dictionary} determines \(X_{\rm TL}\) and, by extension, \(\Delta_{\rm TL}=2^{X_{\rm TL}}\). 

\begin{table}[!htbp]
\centering
\footnotesize
\setlength{\tabcolsep}{4pt}
\renewcommand{\arraystretch}{1.22}
\begin{tabularx}{\textwidth}{@{}
  >{\centering\arraybackslash}p{0.25\textwidth}
  >{\raggedright\arraybackslash}X@{}}
\toprule
Core graph & Full-graph coefficient \(\Gamma_H(a,b,c,d;n)\)\\
\midrule

\CoreFourDictionaryGraph{square \(C_4\)}{1/2,2/3,3/4,4/1}
&
\(\displaystyle
  \Gamma_{C_4}=(n-1)-2(n-1)^3
\)
\\[2mm]

\CoreFourDictionaryGraph{triangle plus leaf}{1/2,2/3,1/3,1/4}
&
\(\displaystyle
\begin{aligned}
  \Gamma_{\rm TL}={}&
  \Gamma_{C_3}(a+d,b,c;n)\bigl(n^d-(n-1)\bigr)
  -\Gamma_{C_3}(a,b,c;n)n^d\\
  &-\mathbf 1_{\substack{n\text{ even},\ a=1,\\ b,c,d\ge2}}
  (n-1)(n-2)\bigl(n^d-2(n-1)\bigr)
\end{aligned}
\)
\\[3mm]

\CoreFourDictionaryGraph{diamond \(K_4\setminus e\)}
  {1/2,2/3,3/4,4/1,1/3}
&
\(\displaystyle
\begin{aligned}
  \Gamma_D={}&-(n-1)^2(3n-1)\\
  &+\mathbf 1_{n\text{ even}}(n-1)(n-2)
  \Bigl[-n\bigl(
  \mathbf 1_{a,b,c\ge2}+\mathbf 1_{a,c,d\ge2}
  \bigr)\\
  &\hspace{23mm}+\mathbf 1_{b,d\ge2}
  +(n-2)\bigl(
  \mathbf 1_{a,b,d\ge2}+\mathbf 1_{b,c,d\ge2}
  \bigr)\\
  &\hspace{23mm}-(3n-5)\mathbf 1_{a,b,c,d\ge2}
  \Bigr]
\end{aligned}
\)
\\[4mm]

\CoreFourDictionaryGraph{complete graph \(K_4\)}
  {1/2,1/3,1/4,2/3,2/4,3/4}
&
\(\displaystyle
\begin{aligned}
  \Gamma_{K_4}={}&-(n-1)(3n^2-3n+1)\\
  &-\mathbf 1_{n\text{ even}}(n-1)(n-2)
  \Bigl[n\bigl(
  \mathbf 1_{a,b,c\ge2}+\mathbf 1_{a,b,d\ge2}\\
  &\hspace{37mm}+\mathbf 1_{a,c,d\ge2}
  +\mathbf 1_{b,c,d\ge2}
  \bigr)\\
  &\hspace{23mm}+3(n-1)\mathbf 1_{a,b,c,d\ge2}
  \Bigr]
\end{aligned}
\)
\\

\bottomrule
\end{tabularx}
\caption{Individual graph coefficients for the four connected four-vertex loop cores, valid for positive GHZ-block sizes \(a,b,c,d\). The triangle coefficient \(\Gamma_{C_3}\) is given in equation~(\ref{trianglegamma}).}
\label{tab:core4-gamma-dictionary}
\end{table}
\FloatBarrier

\leanverifiedbox{The available finite-assignment results underlying Table~1 have been machine-checked for integer $n\geq 2$ and positive block sizes $a,b,c,d$ in the \texttt{Core4loop} project of our public GitHub repository~\cite{StabilizerRenyiLean2026}. The passage from the graph-state replica contraction to the assignment sum, the interpretation of subgraph contributions in equation~\eqref{eq:loop-support-expansion}, and the entropy normalization are supplied in the manuscript and are not part of these Lean declarations.}

\paragraph{Five-vertex and more general loop cores.}

Appendix \ref{app:core5-dictionary} lists the eighteen connected five-vertex loop cores and their full-graph coefficients. The odd-\(n\) formulas cover all eighteen graphs for arbitrary positive GHZ-block sizes. The even-\(n\) expressions are also given there, apart from the explicitly marked subcases of \(G_{16},G_{17},G_{18}\). The appendix also gives a closed formula for the full-graph coefficient of an arbitrary odd cycle \(C_{2h+1}\) for odd \(n\).

\subsection{The six-cycle as a global replica-torus problem}
\label{subsec:c6-global}
We now consider the six-cycle with one qubit in each core block, as shown in Figure~\ref{fig:six-cycle}. Whereas the triangle in Section~\ref{subsec:triangle-main} reduces to a \emph{local} problem involving fewer torus coordinates, the six-cycle leaves a \emph{global} problem on the full replica torus. Its assignment sum remains exactly computable at each integer \(n\ge2\), but we do not obtain a closed expression for arbitrary \(n\).


\begin{figure}[h]
\centering
\begin{tikzpicture}[
    scale=0.8,
    line width=0.9pt,
    vertex/.style={circle, fill=black, inner sep=1.6pt},
    every label/.style={font=\small}
]

\coordinate (v1) at (0,{sqrt(3)});
\coordinate (v2) at (1.5,{sqrt(3)/2});
\coordinate (v3) at (1.5,{-sqrt(3)/2});
\coordinate (v4) at (0,{-sqrt(3)});
\coordinate (v5) at (-1.5,{-sqrt(3)/2});
\coordinate (v6) at (-1.5,{sqrt(3)/2});

\draw (v1) -- (v2) -- (v3) -- (v4) -- (v5) -- (v6) -- cycle;

\node[vertex] at (v1) {};
\node[vertex] at (v2) {};
\node[vertex] at (v3) {};
\node[vertex] at (v4) {};
\node[vertex] at (v5) {};
\node[vertex] at (v6) {};

\node[yshift=10pt] at (v1) {$v_1$};
\node[xshift=10pt,yshift=5pt] at (v2) {$v_2$};
\node[xshift=10pt,yshift=-5pt] at (v3) {$v_3$};
\node[yshift=-10pt] at (v4) {$v_4$};
\node[xshift=-10pt,yshift=-5pt] at (v5) {$v_5$};
\node[xshift=-10pt,yshift=5pt] at (v6) {$v_6$};

\end{tikzpicture}
\caption{Illustration of the six-cycle considered in Section~\ref{subsec:c6-global}.}
\label{fig:six-cycle}
\end{figure}
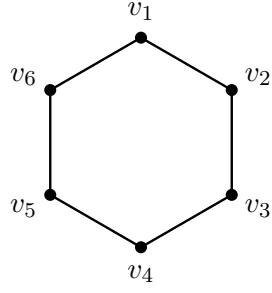


For the simplest instance of the six-cycle problem, consider the graph state
\begin{equation}
  2^{-3}\sum_{x_1,\cdots,x_6\,\in\,\F_2^6}(-1)^{D_{C_6}}
  \ket{x_1x_2x_3x_4x_5x_6}
\end{equation}
with the phase polynomial
\begin{equation}
  D_{C_6}=x_1x_2+x_2x_3+x_3x_4+x_4x_5+x_5x_6+x_6x_1.  
\end{equation}

Choose party \(1\) as reference. The replica torus is $(\Zn)^5$, which we parametrize by $\mathbf{s}=(s_2,\dots,s_6)$. The six qubits in the bras and kets are specified by the fields \(A^{bra},A^{ket},B_2^{bra},B_2^{ket},\ldots,B_6^{bra},B_6^{ket}:(\Zn)^5\rightarrow\F_2\). In this instance, the bra-ket contractions identify $A^{bra}=A^{ket}$ on each torus point and set $B_i^{bra}(\mathbf{s})=B_i^{ket}(\mathbf{s}+\mathbf e_i)$ along each torus direction $2 \leq i \leq 6$. Crucially, the contractions do not further simplify the coordinate dependence of the assignment fields. Writing \(\tau_i\) for the cyclic shift in direction \(i\), the exponent in the $(-1)^N$ phase is 
\begin{align}
N\,=& \phantom{+} \sum_{\mathbf{s}}A\Big((1+\tau_2)B_2 +(1+\tau_6)B_6\Big)\nonumber \\
&+\sum_{\mathbf{s}}\Big(B_2B_3+(\tau_2B_2)(\tau_3B_3)
 +B_3B_4+(\tau_3B_3)(\tau_4B_4)\Big) \label{eq:sixphase} \\
&+\sum_{\mathbf{s}}\Big(B_4B_5+(\tau_4B_4)(\tau_5B_5)
 +B_5B_6+(\tau_5B_5)(\tau_6B_6)\Big), \nonumber
\end{align}
where each assignment field $A,B_2,\dots,B_6$ is a map from $\Zn^5$ to $\F_2$. Note that the three fields $A,B_3, B_5$ each enter $N$ linearly in one line of (\ref{eq:sixphase}). Performing the assignment sum over that trio, we obtain a combined power-of-two factor and three constraints over $B_2,B_4, B_6$. Summing over \(A\) imposes
\begin{equation}\label{eq:6consstraint1}
  (1+\tau_2)B_2(\mathbf{s})+(1+\tau_6)B_6(\mathbf{s})=0  
\end{equation}
for every $\mathbf{s}\in(\Zn)^5$ and contributes \(2^{n^5}\). Similarly, summing over $B_3$ and $B_5$ each gives another \(2^{n^5}\) and imposes:
\begin{equation}\label{eq:6consstraint2}
    \begin{aligned}
        (1+\tau_2\tau_3^{-1})B_2(\mathbf{s})+(1+\tau_3^{-1}\tau_4)B_4(\mathbf{s})&=0,\\
(1+\tau_4\tau_5^{-1})B_4(\mathbf{s})+(1+\tau_5^{-1}\tau_6)B_6(\mathbf{s})&=0.
    \end{aligned}
\end{equation}
Each independent solution of $(B_2,B_4,B_6)$ that satisfies \eqref{eq:6consstraint1} and \eqref{eq:6consstraint2} contributes \(2^{3n^5}\) to the assignment sum. 

It remains to enumerate the linearly independent such solutions. This is equivalent to evaluating the kernel of the operator $M_n:
\left(\F_2^{\Zn^5}\right)^3
\longrightarrow
\left(\F_2^{\Zn^5}\right)^3$, which acts on the triplet of torus fields $(B_2,B_4,B_6)$:
\begin{equation}
    M_n(B_2,B_4,B_6)=
\begin{pmatrix}
(1+\tau_2)B_2+(1+\tau_6)B_6\\
(1+\tau_2\tau_3^{-1})B_2+(1+\tau_3^{-1}\tau_4)B_4\\
(1+\tau_4\tau_5^{-1})B_4+(1+\tau_5^{-1}\tau_6)B_6.
\end{pmatrix}.
\end{equation}
Treating $M_n$ as a \(3n^5\times 3n^5\) matrix over \(\mathbb F_2\), we can therefore write the assignment sum as:
\begin{equation}
\Delta_{C_6}(n)
=
2^{3n^5}\lvert\ker M_n\rvert
=
2^{\,6n^5-\operatorname{rank}_{\mathbb F_2}M_n}.
\label{eq:C6-rank-form}
\end{equation}
We have not been able to obtain an analytic formula for the rank at general $n$. (We return to this point in the Discussion.) The results for $n = 2,3,4$ read:
\begin{equation}
    \begin{array}{c|c|c}
n & \dim_{\F_2}\ker M_n & \log_2\Delta_{C_6}(n)\\
\hline
2 & 34 & 130\\
3 & 157 & 886\\
4 & 498 & 3570
\end{array}.
\end{equation}

\section{Generalization and scope}
\label{sec:generalization-scope}
\label{sec:generalizations}
\label{sec:extensions}

In this section, we show that the finite-field Gauss-sum framework remains applicable when several qubits are grouped into one subsystem and when qubits are replaced by prime-dimensional qudits. We illustrate both extensions with explicit calculations of R\'enyi multi-entropy. 

\subsection{Coarsening the subsystem partition}
\label{subsec:packaged-qubits}

In Sections~\ref{sec:calculus}--\ref{sec:loops}, each physical qubit was treated as a separate subsystem. We now consider coarser partitions, in which some subsystems contain two or more qubits. We consider a \(\mathcal P=\{P_1,\ldots,P_q\}\) on the set of constituent qubits and study \(S_n^{(q)}(P_1:\cdots:P_q)\).

With multi-qubit subsystems, a crucial property of $S_n^{(q)}$ is that it is unchanged under unitary transformations, which act within parties:
\begin{equation}
      \mathcal Z_n(U_{\mathcal P}\ket{\psi})=\mathcal Z_n(\ket{\psi}),\quad S_n^{(q)}(U_{\mathcal P}\ket{\psi})=S_n^{(q)}(\ket{\psi}),\quad \quad U_{\mathcal P}=\bigotimes_{P_\alpha\in\mathcal P}U_{P_\alpha}
\end{equation}
The subsystem-wise unitaries $U_{\mathcal P}$ can significantly simplify the calculation of $S_n^{(q)}$. 

We call a graph state {\bf reducible with respect to $\mathcal{P}$} if a $U_{\mathcal P}$ can transform it to another graph state, which is disconnected. If this is not possible, we call the graph state {\bf irreducible}. A disconected graph state is a tensor product of the graph states, which are associated with its connected components. Because
\begin{equation}
    \mathcal Z_n(\psi_1\otimes\psi_2)
  =\mathcal Z_n(\psi_1)\mathcal Z_n(\psi_2)
  \qquad {\rm and} \qquad
  S_n^{(q)}(\psi_1\otimes\psi_2)
  =S_n^{(q)}(\psi_1)+S_n^{(q)}(\psi_2),
\end{equation}
it is sufficient to calculate the multi-R\'enyi entropies for connected graph states, which are irreducible with respect to the given multi-partition.

In particular, every \(\mathcal{N}\)-qubit stabilizer state with exactly three packaged subsystems is reducible when \(\mathcal{N}\geq4\): the tripartite stabilizer normal-form theorem \cite{Bravyi:2005ztn} states that every pure stabilizer state shared among three subsystems is locally (subsystem-wise) Clifford-equivalent to a tensor product of three-party GHZ states, Bell pairs between packages, and unentangled single-qubit states. Consequently, genuinely irreducible packaged stabilizer states can appear only for four or more subsystems. In addition, all core-2 states (equation~(\ref{eq:twoblockstate}) and Figure~\ref{fig:core2}) are reducible with respect to all multi-partitions except the finest one (one qubit per party).

The calculation still reduces to a finite-field quadratic Gauss sum. Grouping qubits together changes the replica gluing conditions, but they remain linear in the bit labels while the stabilizer phases remain quadratic. We illustrate the calculation with a concrete example below, showing how the R\'enyi multi-entropy responds to the choice of partition.

\paragraph{An irreducible example} The simplest family of genuinely irreducible states comprises core-3 graph states~\eqref{eq:threeblock}, which we discussed in Section~\ref{subsec:worked-examples}. In the text below, we call the state $T$ in order not to clutter the notation.

In expression~\eqref{eq:threeblock}, we work with GHZ block sizes \(a\ge1\) and \(b,c\ge2\).  In this graph state, we choose the multipartition $\mathcal{P}$, in which one qubit from each of the $b$- and $c$-blocks (not the central block) form one two-qubit party \(P\) and all remaining physical qubits are individual parties. Note that with respect to $\mathcal{P}$, the core-3 tree graph is also equivalent to a core-3 triangle graph. A subsystem-local $CZ$ gate acting within $P$ adds an edge between the two corresponding core blocks, which maps the state to the triangle core graph. Therefore:
\begin{equation}
  \Delta_{\triangle,\mathcal{P}}^{a,b,c}(n)
  =\Delta_{T,\mathcal{P}}^{a,b,c}(n)
  \qquad {\rm and} \qquad
  S_{n,\triangle,\mathcal{P}}^{a,b,c}
  =S_{n,T,\mathcal{P}}^{a,b,c}.  
\end{equation}

We now evaluate the R\'enyi multi-entropy $S_n^{q}$ for this multi-partition, where $q=a+b+c-1$. Choosing a qubit from the central GHZ block as reference, the $(q-1)$-dimensional toric lattice is characterized by the coordinates:
\begin{equation}
  \mathbf r\in(\Zn)^{a-1},\qquad
  \mathbf y\in(\Zn)^{b-1},\qquad
  \mathbf z\in(\Zn)^{c-1},\qquad
  p\in\Zn
\end{equation}
Coordinate $P$ is the shared shift direction of the two qubits in $P$. The assignment fields are $A,B,C:\Zn^{a+b+c-2}\rightarrow \F_2$ and satisfy the gluing conditions:
\begin{equation}
    \begin{aligned}
&A^{bra}(\mathbf r,\mathbf y,\mathbf z,p)=A^{ket}(\mathbf r,\mathbf y,\mathbf z,p)\\&A^{bra}(\mathbf r,\mathbf y,\mathbf z,p)=A^{ket}(\mathbf r+\mathbf e_r,\mathbf y,\mathbf z,p)\\
  &B^{bra}(\mathbf r,\mathbf y,\mathbf z,p)=B^{ket}(\mathbf r,\mathbf y+\mathbf e_y,\mathbf z,p)=B^{ket}(\mathbf r,\mathbf y,\mathbf z,p+1)\\
  &C^{bra}(\mathbf r,\mathbf y,\mathbf z,p)=C^{ket}(\mathbf r,\mathbf y,\mathbf z+\mathbf e_z,p)=C^{ket}(\mathbf r,\mathbf y,\mathbf z,p+1),
\end{aligned}
\end{equation}
where \(\mathbf e_r\), \(\mathbf e_y\), and \(\mathbf e_z\) denote arbitrary standard basis vectors in the \(\mathbf r\)-, \(\mathbf y\)-, and \(\mathbf z\)-coordinate sectors, respectively. Let $s:=|\mathbf y|$ and $t:=|\mathbf z|$. The coordinate dependence of these assignment fields can then be written as:
\begin{equation}
    \begin{aligned}
&A^{bra}=A^{ket}=A(\mathbf y,\mathbf z,p)\\
  &B^{bra}=B(\mathbf r,s+p+1,\mathbf z),\ B^{ket}=B(\mathbf r,s+p,\mathbf z)\\
  &C^{bra}=C(\mathbf r,\mathbf y,t+p+1),\ C^{ket}=C(\mathbf r,\mathbf y,t+p).
\end{aligned}
\end{equation}
The two edges of the core graph lead to four terms in the exponent of $(-1)$, which are linear in the root field \(A\). We may first perform the sum over $A$, which yields the constant factor \(2^{n^{b+c-1}}\) and imposes the usual compatibility condition \(K_{BC}=0\), where:
\begin{equation}\label{KBC}
    K_{BC}=\sum_\mathbf r \Big(B(\mathbf r,s+p,\mathbf z)+B(\mathbf r,s+p+1,\mathbf z)+C(\mathbf r,\mathbf y,t+p)+C(\mathbf r,\mathbf y,t+p+1)\Big).
\end{equation}
No residual phase remains after summing over \(A\), so the assignment sum is determined entirely by the number of solutions of \(K_{BC}=0\). 

Let \(\overline B,\overline C\) denote the sums of \(B,C\) over \(\mathbf r\). Using cyclic differences in the combined coordinate directions, the constraint becomes
\begin{equation}
    \mathcal{D}_{s+p} \overline B(s+p,\mathbf z)+\mathcal{D}_{t+p} \overline C(\mathbf y,t+p)=0,
\end{equation}
where $\mathcal{D}_{s+p}=1-\tau_{s+p}$, and $\mathcal{D}_{t+p}=1-\tau_{t+p}$. Mirroring the discussion in Section \ref{subsec:triangle-main}, the invisible sectors of $B,C$ that disappear after the summation over $\mathbf r$ span $(n^{a-1}-1)n^c$ and $(n^{a-1}-1)n^b$ dimensions,  respectively. For the visible sectors, besides the arbitrary $B_0(\mathbf z)$ and $C_0(\mathbf y)$ that vanish under $\mathcal{D}_{s+p}$ and $\mathcal{D}_{t+p}$, we still need to consider the degrees of freedom in $\overline B,\overline C$ for which the two derivatives are nonzero but equal:
\begin{equation}
    \mathcal{D}_{s+p} \overline B(s+p,\mathbf z)=\mathcal{D}_{t+p} \overline C(\mathbf y,t+p).
\end{equation}
The equality must hold for every \(\mathbf y,\mathbf z,p\). By setting \(\mathbf z=\mathbf 0\) and \(\mathbf y=\mathbf 0\), it is easy to check that this remaining sector is parameterized by a single function \(H(s+p+t)\). Thus, the visible sectors are:
\begin{equation}
    \overline B=B_0(\mathbf z)+H(s+p+t),\qquad \overline C=C_0(\mathbf y)+H(s+p+t).
\end{equation}
The free dimensions of \(B_0(\mathbf z)\) and \(C_0(\mathbf y)\) are \(n^{c-1}\) and \(n^{b-1}\), respectively, while that of \(H(s+p+t)\) is \(n-1\), with the constant solution of \(H\) absorbed into \(B_0\) and \(C_0\). Thus, combining the visible and invisible sectors, the solution space of $K_{BC}=0$ has dimension:
\begin{equation}
    (n^{a-1}-1)(n^c+n^b)+n^{c-1}+n^{b-1}+(n-1)
\end{equation}
Together with the constant-factor contribution from the summation over the root field, the assignment sum is:
\begin{equation}
\begin{aligned}
  X_{T,\mathcal{P}}^{a,b,c}(n)
  &=
  n^{b+c-1}+n^{a+c-1}+n^{a+b-1}-\bigl(n^{b-1}+n^{c-1}-1\bigr)(n-1),\\
  \Delta_{T,\mathcal{P}}^{a,b,c}(n)
  &=2^{X_{T,\mathcal{P}}^{a,b,c}(n)} .
\end{aligned}
\label{eq:core3-package-exponent}
\end{equation}

Restoring the factor \(2^{-3}\) from the three-variable core at each of the
\(n^{a+b+c-2}\) torus sites gives:
\begin{equation}
  \mathcal Z_{T,\mathcal{P}}^{a,b,c}(n)
  =2^{X_{T,\mathcal{P}}^{a,b,c}(n)-3n^{a+b+c-2}}.    
\end{equation}
Thus, the multi-R\'enyi entropy for the family of irreducible core-3 trees with a single two-qubit subsystem is:
\begin{equation}
S_{n,T,\mathcal{P}}^{a,b,c}
=\frac{3n^{a+b+c-2}-n^{b+c-1}-n^{a+c-1}-n^{a+b-1}
+(n-1)\bigl(n^{b-1}+n^{c-1}-1)}{(n-1)\,n^{a+b+c-3}}
\log 2
\label{eq:core3-package-entropy}
\end{equation}
Within this family, the smallest irreducible example is \((a,b,c)=(1,2,2)\), for which:
\begin{equation}
  X_{T,\mathcal{P}}^{1,2,2}(n)=n^3+3n-1
  \qquad {\rm and} \qquad
  S_{n,T,\mathcal{P}}^{1,2,2}
  =\left(2+\frac{2}{n}-\frac{1}{n^2}\right) \log 2.
  \label{eq:core3-package-minimal}
\end{equation}

\subsection{Changing the coefficient field: Prime qudits}
\label{subsec:prime-qudits}


We now extend the calculation from qubits to prime-dimensional qudits. This replaces the binary field \(\F_2\) with \(\Fp\)~\cite{gross2007hudson}. The finite-field Gauss sum method remains applicable and the resulting R\'enyi multi-entropies closely resemble those of the corresponding qubit graphs. As we shall see, the qudit calculations closely resemble their qubit predecessors. An obvious replacement is $\log 2 \to \log p$ in all logarithmic quantities.


Let $p$ be prime and $\omega_p=e^{2\pi i/p}$. Qudit graph states~\cite{keet2010qudit} are defined by edge-weighted graphs whose edges encode applications of the weighted gate
\begin{equation}
    CZ_{ij}^{(w)}\ket{x_i x_j}
  =\omega_p^{w x_i x_j}\ket{x_i x_j}
  \qquad {\rm for}~~x_i,x_j\in\Fp,\; w\in\Fp^\times.
\end{equation}
A weighted graph is a triple $G=(V,E,w)$ where $w:E\to\Fp^\times$ assigns a nonzero weight to each edge. In comparison to the qubit discussion, the phase polynomial $D_G(x)=\sum_{(u,v)\in E}w_{uv}x_u x_v$ now sets the exponent of $\omega_p$. 

By a local Clifford scaling at the leaves, any prime-dimensional qudit graph state can be brought to a form whose leaf
edges all have unit weight~\cite{BahramgiriBeigi2007,RevisZakaryanRaissi2025}. Therefore, in defining a core graph representation of a qudit state, we do not have to carry the information about the weights of the leaves; we can assume they are all set to unity. With this assumption, mirroring the definition of the qubit core graph representation, we set:

\begin{definition}
\label{def:prime-core-graph-state}
Let $G=(V,E,w)$ be a weighted graph and $\mathbf a\in\mathbb Z_{>0}^{|V|}$ a block-size vector.  {\bf The prime-field core-graph state} $\ket{\Psi_{G,\mathbf a}^{(p)}}$ in the computational basis is:
\begin{equation}\label{eq:prime-core-state}
\ket{\Psi_{G,\mathbf a}^{(p)}}
:=p^{-k/2}\sum_{x\,\in\,\Fp^k}\omega_p^{D_G(x)}
\bigotimes_{v\in V}\ket{x_v^{a_v}},
\qquad
D_G(x)=\sum_{(u,v)\,\in\, E}w_{uv}x_u x_v,
\end{equation}
where $k:=|V|$ and $\ket{x_v^{a_v}}$ records $a_v$ repeated copies of the same $p$-dit label.
\end{definition}
\noindent A physical graph state is related to its core graph by a local Clifford transformation, which a qudit-wise Fourier transform:
\begin{equation}
    \ket{G_{\rm phys}^{(p)}}
    =
    \left(
        \bigotimes_{v\in V}
        \bigotimes_{r\in L_v} F_{p,r}
    \right)
    \ket{\Psi_{G,\mathbf a}^{(p)}}
    \qquad {\rm where} \qquad
        F_p\ket{x}
    =
    p^{-1/2}\sum_{y\,\in\,\mathbb F_p}\omega_p^{xy}\ket{y}
\end{equation}

Let us highlight the differences between $p=2$ and general $p$-dimensional qudits:
\begin{itemize}
    \item The normalization of the state: $2^{-k/2} \to p^{-k/2}$.
    \item The phase: $(-1)^{D_G(x)} \to \omega_p^{D_G(x)}$.
    \item The exponent now involves the difference between ket and bra phases because of the complex conjugation in the bra states. (At $p=2$, there was no difference between addition and subtraction.)
    \item Variables that enter $D_G(x)$ linearly are summed using:
    \begin{equation}
    \sum_{a\,\in\,\Fp}\omega_p^{ak}=p\,\delta_{k,0}.
\end{equation}
Once again, summing over variables that appear linearly in the exponent imposes linear constraints. After solving them, the phase is at most quadratic in the remaining (unconstrained) assignment fields.
    \item For loop graphs, the distinction between odd and even $n$ branches lifts to a distinction between $p\nmid n$ and $p \mid n$. We illustrate this below.
\end{itemize}
\noindent
For trees, equation~\eqref{eq:tree-exponent} still applies, now understood as the base-$p$ exponent: $\Delta_T^{(p)}(n)=p^{X_T(n)}$. The argument in Section~\ref{subsec:subtreeformula} is only potentially affected by going from $p=2$ to general $p$ in the leaf elimination formula. However, the key quantity there---the rank of the cyclic difference operator $w(1-\tau)$, equal to $(n-1)$---is unaffected by going to other prime $p$ because the kernel still comprises only constant functions.

\paragraph{A loop example: Two distinct branches}
For a comparison with Section~\ref{subsec:triangle-main}, we consider the triangular core graph. The multi-R\'enyis again develop two branches, now depending on whether $p$ divides $n$. For the unit-weight triangle with two physical qudits in each core block:
\begin{equation}
\Delta_{\triangle,(2,2,2)}^{(p)}(n)=
\begin{cases}
p^{3n^4-3n^3+5n^2-3n+1}& ~{\rm for}~~p\nmid n,\\[2mm]
p^{3n^4-3n^3+6n^2-6n+3}& ~{\rm for}~~p\mid n.
\end{cases}    
\label{genploop}
\end{equation}
In other words, the exponents are the same as in equations~(\ref{oddnbranch}-\ref{evennbranch}). For arbitrary block sizes $a,b,c$, the exponent (now in base $p$) in equation~\eqref{abclog} still applies. 
As before, the exceptional branch (now for $p \mid n$) still requires $a,b,c\ge2$. 

To see the arithmetic origin of the branch, let:
\begin{equation}
  \langle f,g\rangle=\sum_{t\,\in\,\Zn}f(t)\,g(t)\in\Fp
\end{equation}
Note that for the constant function $\one=(1,\ldots,1)$, we have
  $\langle\one,\one\rangle=0 \in \mathbb{F}_p$ if $p \mid n$ and non-zero otherwise.
This precisely parallels the role of the constant function in the binary triangle calculation of Section~\ref{subsec:triangle-main}. 

\paragraph{Analogue of Subsection~\ref{subsec:c6-global}}
In calculating the multi-R\'enyi entropies for qudit stabilizer states, we often encounter the same difficulty as in Section~\ref{subsec:c6-global}, that is  residual fields that cannot be eliminated locally on the global replica torus. In such cases, the calculation involves a matrix rank, which is challenging to compute analytically for all $n$. The simplest example where this happens is the 
four-qutrit AME\((4,3)\) state~\cite{helwig2013ame}
\begin{equation}
    \ket{\Omega}
  =\frac{1}{3}\sum_{a,s\in\mathbb F_3}
    \ket{a}\otimes \ket{ a+s}\otimes \ket{s}\otimes \ket{s-a},
\end{equation}
which is also called a perfect tensor in some literature \cite{Kbasis2019}. 
\section{Discussion and outlook}
\label{sec:discussion}

We have calculated the multi-R\'enyi entropies for a broad class of stabilizer states. The purpose of this exercise is to supply cases studies, from which one might garner qualitative and quantitative insights about the behavior of the multi-R\'enyis. Our key results are formula~(\ref{eq:tree-exponent}), which covers all tree graph states, and formula~(\ref{eq:odd-bridge-factorization}) along with an assortment of individual loop graph coefficients $\Gamma_L$, which pertains to graph states with loops. We have also observed that the multi-R\'enyi entropies of graph states with loops have qualitatively distinct behaviors at even and odd $n$ (Section~\ref{subsec:triangle-main}), and identified a quintessentialy non-local aspect of the computation, which is showcased by the 6-cycle graph state (Section~\ref{subsec:c6-global}).

The last two points might hint at some qualitative insights. The distinction between even and odd replica index seems to be a signature of graph states with loops. Moreover, its lift to $p$-dimensional qudits (Section~\ref{subsec:prime-qudits}) could be harnessed to detecting $p$. If we had an operational way to control integer-indexed multi-R\'enyi entropies of unknown states, we could use the distinction in equation~(\ref{genploop}) to determine the base $p$.

The non-local aspects of the computation encountered in Section~\ref{subsec:c6-global} are also intriguing. A preliminary study of the 6-cycle graph state's multi-R\'enyi entropies at larger $n$ indicates a non-polynomial behavior with some exotic ingredients that depend on the factors of $n$. The key source of the difficulty in this computation is that the linear conditions obtained after summing over the linearly appearing assignment fields are non-local in character. This non-locality is reminiscent of other features of the 6-cycle graph state, which have found use elsewhere \cite{WRC2024}. Indeed, this quantum state is the paradigmatic resource state for Measurement-Based Quantum Computation (MBQC) \cite{raussendorf2001oneway}. For structural reasons, we expect the universal MBQC resource states such as the two-dimensional cluster state \cite{raussendorf2001oneway,markov2008simulating,Nest:2006rge} to exhibit the same type non-locality in its multi-R\'enyi spectrum. It would be fascinating to understand if the said behavior of the multi-R\'enyis is conceptually linked to the computational power possessed by those states.

Most of the analysis in this paper concerns the finest partition of constituent qubits: one qubit per party. We briefly sketched the generalization to coarser partitions in Section~\ref{subsec:packaged-qubits}. We have not attempted a systematic study of just how much information about the state one gets by studying the complete set of the multi-R\'enyi entropies under all multi-partitions. To illustrate this question, consider the example in Figure~\ref{fig:tree-counterexample}. It shows two graph states, which have identical multi-R\'enyi spectra {\it for the finest partition} (one qubit per party). Does there exist a pair of states, for which all multi-R\'enyi entropies over {\it all} partitions are identical? If so, what is such a minimal example? We leave this question to future work.

\paragraph{Limits and monotonicity.}
It is interesting to study the \(n\to\infty\) and \(n\to1\) limits of the R\'enyi multi-entropies. We collect some representative examples in Table~\ref{tab:limit-values}. In every qubit stabilizer family solved here, \(S_\infty^{(q)}\) is an integer multiple of $\log 2$. 

For unit-block trees:
\begin{equation}
    \frac{S_\infty^{(q=\mathcal{N})}(T)}{\log 2}
  =
  \mathcal{N}-\sum_{\substack{
      \varnothing\neq U\subseteq V(T)\\
      T[U]\ \mathrm{connected}
  }}(-1)^{|U|-1}.
\end{equation}
This limit turns out to reproduce a quantity of independent interest: the maximum min-entropy over all bipartitions of the full system. Preliminary studies indicate that this conclusion persists for all stabilizer states, including those whose graph representations involve loops. We will investigate (and hopefully prove) this relation in forthcoming work~\cite{CzechFengJuWuZhao2026}.

The \(n\to1\) limit, which motivated the introduction of the multi-R\'enyi entropies in holographic duality, is---as usual---marred by an ambiguity concerning the choice of analytic continuation. In Table~\ref{tab:limit-values} we display the most na{\"\i}ve extrapolations. 


Unlike ordinary R\'enyi entropy~\cite{Zyczkowski2005}, the polynomial continuations of R\'enyi multi-entropy need not decrease monotonically for real \(n>0\). Restricting instead to integer \(n\ge2\), where the replica contractions can be evaluated directly without analytic continuation, \(S_n^{(q)}\) appears to be nonincreasing in stabilizer examples examined here. A general proof of this integer-\(n\) monotonicity remains open~\cite{GaddeJainKrishnaKulkarniSharma2023}.

\begin{table}[!htbp]
\centering
\small
\setlength{\tabcolsep}{5pt}
\renewcommand{\arraystretch}{1.25}
\begin{tabularx}{\textwidth}{
  @{}
  >{\raggedright\arraybackslash}p{0.36\textwidth}
  >{\raggedright\arraybackslash}X
  cc
  @{}
}
\toprule
Family (parties)
& \(S_n/\log 2\)
& \(n\to1\)
& \(n\to\infty\) \\
\midrule
Bell pair (\(q=2\))
& \(1\)
& \(1\)
& \(1\) \\

Three-party star = GHZ (\(q=3\))
& \(1+\frac1n\)
& \(2\)
& \(1\) \\

Star with block sizes \((1,1,2)\) (\(q=4\))
& \(2+\frac1{n^2}\)
& \(3\)
& \(2\) \\

Star with block sizes \((1,2,2)\), un-packaged (\(q=5\))
& \(2+\frac2n-\frac1{n^2}+\frac1{n^3}\)
& \(4\)
& \(2\) \\

Same state, packaged (\(q=4\))
[Eq.~\eqref{eq:core3-package-minimal}]
& \(2+\frac2n-\frac1{n^2}\)
& \(3\)
& \(2\) \\

\(6\)-qubit triangle, odd \(n\) (\(q=6\))
& \(3+\frac3{n^2}-\frac2{n^3}+\frac1{n^4}\)
& \(5\)
& \(3\) \\

\(6\)-qubit triangle, even \(n\) (\(q=6\))
& \(3+\frac3{n^2}-\frac3{n^3}+\frac3{n^4}\)
& \(6\)
& \(3\) \\
\bottomrule
\end{tabularx}
\caption{
Representative formulas and limit values, in units of \(\log 2\). The \(n\to1\) entries are extrapolations of the displayed formulas. Both the triangle branches and the finest / minimally coarse partitions agree through order \(1/n^2\) and differ at order \(1/n^3\). Each pair has different \(n\to1\) values but a common large-\(n\) limit. Results for the minimally coarse partition (two qubits in one party and all others in individual parties) are indicated as `packaged' in the table.}
\label{tab:limit-values}
\end{table}
\FloatBarrier

\paragraph{Quantifying magic?}
The progress we have made in computing the multi-R\'enyi entropies relies heavily on using stabilizer states. We can turn this upside down and ask: Can we detect magic (non-stabilizerness) of the state by noting a departure of its multi-R\'enyi entropies from the formulas derived here? More ambitiously, can we design a measure of magic, which is built up from the multi-R\'enyis?

We leave the latter question open. For the former, the property that all the multi-R\'enyi entropies are rational multiples of $\log 2$ is a potential witness of non-stabilizerness. The origin of this phenomenon is the quadratic form of the phase exponent in stabilizer states. As exemplified in Section~\ref{subsec:triangle-main}, quadratic exponents naturally organize the computation into two parts: (1) Linearly appearing assignment field impose linear constraints and the solution space has a dimension, which is a power of 2. (2) Each hyperbolic pair (residual bilinearly appearing fields) cuts that dimension by a factor of $1/2$---still a power of 2. As soon as the phase is not quadratic, the assignment sum has no reason to be a power of 2 and its logarithm---the entropy---fails to be an integer multiple of $\log 2$.

A simple example is the three-qubit hypergraph
state~\cite{Rossi2013Hypergraph}:
\begin{equation}
\begin{aligned}
  \ket{\psi_{xyz}}
  &
    =\frac{1}{\sqrt8}
      \sum_{x,y,z\,\in\,\F_2}(-1)^{xyz}\ket{xyz}\\
  &=\frac{1}{\sqrt8}\bigl(
      \ket{000}+\ket{001}+\ket{010}+\ket{011}
      +\ket{100}+\ket{101}+\ket{110}-\ket{111}
    \bigr).
\end{aligned}
\label{eq:cubic-state}
\end{equation}
Crucially, its phase polynomial $D(x,y,z) = xyz$ is cubic.



We have evaluated the assignment sums for replica indices $n=2, 3$ by hand:
\begin{equation}
\begin{aligned}
  \Delta_{xyz}(n=2)&=2^5\cdot29
  &\qquad &
  S_2^{(3)}(\psi_{xyz})
    =\frac{7-\log_2 29}{2} \log 2 \\
  \Delta_{xyz}(n=3)&=2^{10}\cdot3917
  &\qquad & 
  S_3^{(3)}(\psi_{xyz})
    =\frac{17-\log_2 3917}{6} \log 2
\end{aligned}
\label{eq:cubic-small-n}
\end{equation}
The coefficients of $\log 2$ are irrational.

\paragraph{Proof status.}
The Lean~4 source code and build documentation for the finite-assignment results underlying the two- and three-block formulas (Sec.~2.4), the tree graph theorem (Sec.~3), the triangle core formula (Sec.~4.2), and the four-vertex core results (Sec.~4.3) are publicly available in our GitHub repository~\cite{StabilizerRenyiLean2026}. The four-vertex core results (Sec. \ref{subsec:core4-main}) are presented without detailed proofs, with supporting Lean source code and build instructions are publicly available in our GitHub repository~\cite{StabilizerRenyiLean2026}. The composition rule (Sec. \ref{subsec:universal-known-graph-formula}) and the coefficients for five-vertex and more general loop cores (Appendix \ref{app:core5-dictionary}) are stated without proof. Supporting documentation, which establishes these results is available upon request.

\paragraph{Use of large language models.}
Most formulas were proposed or conjectured by the authors, who also provided the main ideas for their proofs, with ChatGPT assisting in the derivations. Exceptions include the proofs of the four- and five-vertex core results and of the general odd-cycle formula, which mainly relied on ChatGPT. We used models from the GPT-5 and GPT-6 families. The Lean formalizations relied heavily on GPT-5 models. All proofs were subsequently independently checked using GPT-6 models.

\section*{Acknowledgments}
We are grateful to Abhijit Gadde and Sriram Akella for useful discussions and for synchronizing the arXiv submissions of our papers. YZ thanks Prof. Ya-Wen Sun for support and useful discussions. BC thanks the organizers of the workshops ``Discussions on Quantum Spacetime'' held at LMSI (Mumbai, India)
and ``Quantum Gravity, Chaos, and Observers for Black Holes and Cosmology'' held at the Aspen Center for Physics, where some of this work was completed. The work of BC, YF, and XXJ is supported by an NSFC grant number 12042505. The work of JH is supported by MEXT KAKENHI Grant-in-Aid for Transformative Research Areas (A) through the ``Extreme Universe'' collaboration: Grant Number 21H05187. The work of XXJ is also supported by the Shuimu Tsinghua Scholar Program of Tsinghua University.

\appendix

\section{Dictionary of five-vertex and more general loop cores}
\label{app:core5-dictionary}

This appendix collects the full-graph coefficients \(\Gamma_H(n)\) for the eighteen connected five-vertex core graphs with loops in Table~\ref{tab:core5-gamma-dictionary}. The final row also gives two general formulas for odd \(n\): the product rule for two loop cores sharing exactly one vertex, and the closed expression for an arbitrary odd cycle \(C_{2h+1}\).

The vertex labels \(a,b,c,d,e\) denote positive block sizes. For each graph \(G_i\), \(B_i(n)\) denotes the polynomial in its odd-\(n\) entry. Set \(\delta_n=(n-1)(n-2)\), \(t_x=\mathbf{1}_{\{x\ge2\}}\), and \(t_S=\prod_{x\in S}t_x\); for example, \(t_{abc}=t_at_bt_c\).

The functions \(N_5,p,q,r\) and the correction factors \(F_7,\ldots,F_{18}\) used in the even-\(n\) entries are defined under ``Auxiliary functions'' after the table. The entries marked with a dagger exclude the even-\(n\) subcase \(a,b,c,d\geq2\), \(e=1\).

In the final row, \(C_1\vee_v C_2\) denotes the graph obtained by identifying one vertex of \(C_1\) with one vertex of \(C_2\), calling the shared vertex \(v\), and retaining all edges of both graphs.

\begingroup

\newcommand{\CoreFiveSharedLoopCores}{%
  \begin{tikzpicture}[
    x=7.8mm,y=7.8mm,
    baseline=(current bounding box.center),
    line width=.45pt,
    every node/.style={font=\scriptsize,inner sep=0pt}]
    \coordinate (v)  at (0,0);
    \coordinate (a1) at (-.58,.43);
    \coordinate (b1) at (-.43,1.02);
    \coordinate (c1) at (.43,1.02);
    \coordinate (d1) at (.58,.43);
    \coordinate (a2) at (-.58,-.43);
    \coordinate (b2) at (-.43,-1.02);
    \coordinate (c2) at (.43,-1.02);
    \coordinate (d2) at (.58,-.43);

    \draw (v)--(a1)--(b1);
    \draw (c1)--(d1)--(v);
    \draw (v)--(a2)--(b2);
    \draw (c2)--(d2)--(v);

    \foreach \p in {a1,b1,c1,d1,a2,b2,c2,d2}
      \fill (\p) circle[radius=1pt];
    \fill (v) circle[radius=1.3pt];

    \node at (0,1.10) {$\cdots$};
    \node at (0,-1.10) {$\cdots$};
    \node at (0,.58) {$C_1$};
    \node at (0,-.58) {$C_2$};
    \node[anchor=west] at (.14,0) {$v$};
  \end{tikzpicture}%
}

\newcommand{\CoreFiveOddCycle}{%
  \begin{tikzpicture}[
    x=8mm,y=8mm,
    baseline=(current bounding box.center),
    line width=.45pt,
    every node/.style={font=\scriptsize,inner sep=0pt}]
    \coordinate (a) at (0,.82);
    \coordinate (b) at (.72,.38);
    \coordinate (c) at (.66,-.43);
    \coordinate (d) at (-.66,-.43);
    \coordinate (e) at (-.72,.38);

    \draw (d)--(e)--(a)--(b)--(c);
    \draw (c)--(.30,-.69);
    \draw (-.30,-.69)--(d);

    \foreach \p in {a,b,c,d,e}
      \fill (\p) circle[radius=1pt];

    \node at (0,-.70) {$\cdots$};
    \node at (0,-1.16) {$C_{2h+1}$};
  \end{tikzpicture}%
}

\fontsize{9}{10.5}\selectfont
\setlength{\tabcolsep}{2pt}
\renewcommand{\arraystretch}{1}

\setlength{\LTleft}{\fill}
\setlength{\LTright}{\fill}
\setlength{\LTcapwidth}{\textwidth}
\setlength{\LTpre}{6pt}
\setlength{\LTpost}{6pt}

\newpage
\begin{longtable}{@{}
  >{\centering\arraybackslash}m{0.12\textwidth}
  >{\raggedright\arraybackslash}m{0.36\textwidth}
  >{\centering\arraybackslash}m{0.12\textwidth}
  >{\raggedright\arraybackslash}m{0.36\textwidth}@{}}

\caption{Coefficients for the eighteen connected five-vertex graphs
with loops, followed by two general identities for odd \(n\).}
\label{tab:core5-gamma-dictionary}\\
\toprule
Graph & \(\Gamma_H(n)\) & Graph & \(\Gamma_H(n)\)\\
\midrule
\endfirsthead

\caption[]{Coefficients \(\Gamma_H(n)\) (continued).}\\
\toprule
Graph & \(\Gamma_H(n)\) & Graph & \(\Gamma_H(n)\)\\
\midrule
\endhead

\bottomrule
\endfoot

\bottomrule
\addlinespace[2pt]
\multicolumn{4}{@{}p{\dimexpr .96\textwidth+6\tabcolsep\relax}@{}}{%
  \footnotesize
  \textsuperscript{\(\dagger\)}\,For \(G_{16},G_{17},G_{18}\), only the
  non-polynomial even-\(n\) subcase \(a,b,c,d\ge2,\ e=1\) is omitted.
}\\
\endlastfoot


\CoreFiveDictionaryGraph{G_1}{1/2,2/3,3/4,4/5,5/1}
& \(\begin{aligned}
\text{odd }n:\;& 3n^4-10n^3+10n^2-3,\\[1pt]
\text{even }n:\;& B_{1}(n)+\delta_n \bigl((n-1)^2+1\bigr)\mathcal N_5
\end{aligned}\)
& \CoreFiveDictionaryGraph{G_2}{1/2,2/3,3/4,4/1,1/5}
& \(\begin{aligned}
\text{odd }n:\;& 2(n-1)^4-(n-1)^2,\\[1pt]
\text{even }n:\;& B_{2}(n)
\end{aligned}\)\\[0.5mm]

\CoreFiveDictionaryGraph{G_3}{1/2,2/3,3/1,1/4,4/5}
& \(\begin{aligned}
\text{odd }n:\;& (n-1)^3(2n-1),\\[1pt]
\text{even }n:\;& B_{3}(n)+\delta_n t_{bc}\,q(a;d,e)
\end{aligned}\)
& \CoreFiveDictionaryGraph{G_4}{1/2,2/3,3/1,1/4,1/5}
& \(\begin{aligned}
\text{odd }n:\;& (n-1)^3(2n-1),\\[1pt]
\text{even }n:\;& B_{4}(n)+\delta_n t_{bc}\,r(a;d,e)
\end{aligned}\)\\[0.5mm]

\CoreFiveDictionaryGraph{G_5}{1/2,2/3,3/1,1/4,2/5}
& \(\begin{aligned}
\text{odd }n:\;& (n-1)^3(2n-1),\\[1pt]
\text{even }n:\;& B_{5}(n)+\delta_n t_c\,p(a;d)p(b;e)
\end{aligned}\)
& \CoreFiveDictionaryGraph{G_6}{1/4,1/5,2/4,2/5,3/4,3/5}
& \(\begin{aligned}
\text{odd }n:\;& 3n^4-11n^3+13n^2-6n+1,\\[1pt]
\text{even }n:\;& B_{6}(n)
\end{aligned}\)\\[0.5mm]

\CoreFiveDictionaryGraph{G_7}{1/5,2/4,2/5,3/4,3/5,4/5}
& \(\begin{aligned}
\text{odd }n:\;& (n-1)^3(3n-1),\\[1pt]
\text{even }n:\;& B_{7}(n)+\delta_n \mathcal F_{7}
\end{aligned}\)
& \CoreFiveDictionaryGraph{G_8}{1/3,1/4,1/5,2/4,3/5,4/5}
& \(\begin{aligned}
\text{odd }n:\;& (n-1)^3(3n-1),\\[1pt]
\text{even }n:\;& B_{8}(n)+\delta_n \mathcal F_{8}
\end{aligned}\)\\[0.5mm]

\CoreFiveDictionaryGraph{G_9}{1/3,1/4,1/5,2/4,2/5,3/5}
& \(\begin{aligned}
\text{odd }n:\;& 4n^4-13n^3+12n^2-3,\\[1pt]
\text{even }n:\;& B_{9}(n)+\delta_n \mathcal F_{9}
\end{aligned}\)
& \CoreFiveDictionaryGraph{G_{10}}{1/5,2/3,2/4,2/5,3/4,3/5,4/5}
& \(\begin{aligned}
\text{odd }n:\;& (n-1)^2(3n^2-3n+1),\\[1pt]
\text{even }n:\;& B_{10}(n)+\delta_n \mathcal F_{10}
\end{aligned}\)\\[0.5mm]

\CoreFiveDictionaryGraph{G_{11}}{1/3,1/4,1/5,2/4,2/5,3/4,3/5}
& \(\begin{aligned}
\text{odd }n:\;& (n-1)(5n^3-11n^2+6n-1),\\[1pt]
\text{even }n:\;& B_{11}(n)+\delta_n \mathcal F_{11}
\end{aligned}\)
& \CoreFiveDictionaryGraph{G_{12}}{1/4,1/5,2/4,2/5,3/4,3/5,4/5}
& \(\begin{aligned}
\text{odd }n:\;& (n-1)^2(2n-1)^2,\\[1pt]
\text{even }n:\;& B_{12}(n)+\delta_n \mathcal F_{12}
\end{aligned}\)\\[0.5mm]

\CoreFiveDictionaryGraph{G_{13}}{1/3,1/4,1/5,2/4,2/5,3/4,3/5,4/5}
& \(\begin{aligned}
\text{odd }n:\;& (n-1)^2(5n^2-5n+1),\\[1pt]
\text{even }n:\;& B_{13}(n)+\delta_n \mathcal F_{13}
\end{aligned}\)
& \CoreFiveDictionaryGraph{G_{14}}{1/3,1/4,1/5,2/3,2/4,2/5,3/4,3/5,4/5}
& \(\begin{aligned}
\text{odd }n:\;& (n-1)^2(5n^2-4n+1),\\[1pt]
\text{even }n:\;& B_{14}(n)+\delta_n \mathcal F_{14}
\end{aligned}\)\\[0.5mm]



\CoreFiveDictionaryGraph{G_{15}}{1/2,1/3,1/4,1/5,2/3,2/4,2/5,3/4,3/5,4/5}
& \(\begin{aligned}
\text{odd }n:\;& (n-1)(2n-1)(2n^2-2n+1),\\[1pt]
\text{even }n:\;& B_{15}(n)+\delta_n \mathcal F_{15}
\end{aligned}\)
& \CoreFiveDictionaryGraph{G_{16}}{1/3,1/4,1/5,2/4,2/5,3/5,4/5}
& \(\begin{aligned}
\text{odd }n:\;& (n-1)(5n^3-10n^2+3n+3),\\[1pt]
\text{even }n:\;& B_{16}(n)+\delta_n \mathcal F_{16}\,{}^\dagger
\end{aligned}\)\\[0.5mm]

\CoreFiveDictionaryGraph{G_{17}}{1/3,1/5,2/4,2/5,3/5,4/5}
& \(\begin{aligned}
\text{odd }n:\;& (n-1)^2(2n-1)^2,\\[1pt]
\text{even }n:\;& B_{17}(n)+\delta_n \mathcal F_{17}\,{}^\dagger
\end{aligned}\)
& \CoreFiveDictionaryGraph{G_{18}}{1/3,1/4,1/5,2/3,2/4,2/5,3/5,4/5}
& \(\begin{aligned}
\text{odd }n:\;& (n-1)(6n^3-12n^2+6n-1),\\[1pt]
\text{even }n:\;& B_{18}(n)+\delta_n \mathcal F_{18}\,{}^\dagger
\end{aligned}\)\\[0.5mm]

\midrule

\CoreFiveSharedLoopCores
& \(\begin{aligned}
  &\text{odd }n:\\[2pt]
  &\Gamma_{C_1\vee_v C_2}(n)
   =\Gamma_{C_1}(n)\Gamma_{C_2}(n),\\[2pt]
  &C_1\cap C_2=\{v\}.
\end{aligned}\)
& \CoreFiveOddCycle
& \(\text{odd }n,\quad h\in\mathbb Z_{\ge1}:\)\par\vspace{2pt}
  \(\begin{aligned}
  \Gamma_{C_{2h+1}}(n)
  &=(h+1)(n-1)^{2h}\\[1pt]
  &\quad+h\,\frac{(n-1)^{2h}+n-1}{n}.
\end{aligned}\)\\[1mm]

\end{longtable}
\endgroup

\begingroup
\small
\setlength{\jot}{1pt}
\setlength{\abovedisplayskip}{3pt}
\setlength{\belowdisplayskip}{3pt}
\setlength{\abovedisplayshortskip}{3pt}
\setlength{\belowdisplayshortskip}{3pt}
\allowdisplaybreaks[1]
\noindent\textbf{Auxiliary functions.}
Set
\[
u_x=1+(n-2)t_x,\qquad
p(x;y)=(1-t_x)u_y-(n-1)t_x.
\]
For \(\mathcal N_5\), put \((t_1,t_2,t_3,t_4,t_5)=(t_a,t_b,t_c,t_d,t_e)\),
with indices read modulo five:
\[
\mathcal N_5=
\sum_{i=1}^{5}t_it_{i+1}t_{i+2}
-\sum_{i=1}^{5}t_it_{i+1}t_{i+2}t_{i+3}
+2t_1t_2t_3t_4t_5.
\]
The functions in the \(G_3\) and \(G_4\) entries are
\[
q(a;d,e)=
\begin{cases}
(n-1)^2,&a\ge2,\\
-(n-1)^2,&a=1,\ d\ge2,\\
-1,&a=d=e=1,\\
(n-1)(n-3),&a=d=1,\ e\ge2,
\end{cases}
\qquad
r(a;d,e)=
\begin{cases}
(n-1)^2,&a\ge2,\\
-\bigl[u_du_e+(n-1)(u_d+u_e)\bigr],&a=1.
\end{cases}
\]
The remaining correction factors are
\begin{align*}
\mathcal F_{7}={}& -t_{bc}-n(t_{bd}+t_{cd})+(n-2)^2t_{abc}-n(n-2)(t_{abd}+t_{acd})-(n^2-3)t_{bcd}-n(n-2)t_{bce}\\*
&{}+n^2(t_{bde}+t_{cde})-(n-2)(3n-5)t_{abcd}-(n-2)^2t_{abce}+n(n-2)(t_{abde}+t_{acde})+n(3n-5)t_{bcde}\\*
&{}+(n-2)(3n-5)t_{abcde},\\[3pt]
\mathcal F_{8}={}& t_c+(n-2)(t_{ac}+t_{bc}+t_{ce})-n(t_{ae}+t_{cd})+(n-2)^2t_{abc}-n(n-2)(t_{abe}+t_{acd})\\*
&{}+(n^2-4n+5)t_{ace}+n^2t_{ade}-(n-2)t_{bcd}+(n-2)^2t_{bce}-n(n-2)t_{cde}-(n-2)^2t_{abcd}\\*
&{}-(n-2)(3n-5)t_{abce}+n(n-2)t_{abde}+n(3n-5)t_{acde}-(n-2)^2t_{bcde}+(n-2)(3n-5)t_{abcde},\\[3pt]
\mathcal F_{9}={}& -n(t_{ac}+t_{ce})-(n^2-2n-1)(t_{abc}+t_{cde})+(n^2-2n+2)(t_{abd}+t_{bde})+(n^2-4n+5)(t_{acd}+t_{bce})\\*
&{}+2(n^2-n+1)t_{ace}+(n-2)^2t_{bcd}+(-3n^2+11n-12)(t_{abcd}+t_{bcde})+2(n-3)(t_{abce}+t_{acde})\\*
&{}-(n^2-2n+2)t_{abde}+(5n^2-18n+20)t_{abcde},\\[3pt]
\mathcal F_{10}={}& -n(t_{bc}+t_{bd}+t_{cd})-n(n-2)(t_{abc}+t_{abd}+t_{acd})+(n-1)(n-3)t_{bcd}+n^2(t_{bce}+t_{bde}+t_{cde})\\*
&{}-3(n-1)(n-2)t_{abcd}+n(n-2)(t_{abce}+t_{abde}+t_{acde})+3n(n-1)t_{bcde}+3(n-1)(n-2)t_{abcde},\\[3pt]
\mathcal F_{11}={}& t_b+(n-2)(t_{ab}+t_{bc}+t_{bd}+t_{be})-(n+1)t_{ac}-nt_{de}+(-2n^2+2n+5)t_{abc}\\*
&{}+(n-2)^2(t_{abd}+t_{bcd}+t_{abe}+t_{bce})+\bigl[n^2+(n-2)^2\bigr](t_{acd}+t_{ace})-n(n-2)(t_{ade}+t_{cde})\\*
&{}+\bigl[(n-2)^2+1\bigr]t_{bde}-(n-2)(2n-7)(t_{abcd}+t_{abce})-(n-2)(3n-5)(t_{abde}+t_{bcde})\\*
&{}+(2n^2+n-8)t_{acde}+(9n^2-34n+33)t_{abcde},\\[3pt]
\mathcal F_{12}={}& -n(t_{ab}+t_{ac}+t_{bc})+(n-1)(n-3)t_{abc}-n(n-2)(t_{abd}+t_{acd}+t_{bcd}+t_{abe}+t_{ace}+t_{bce})\\*
&{}+n^2(t_{ade}+t_{bde}+t_{cde})-3(n-1)(n-2)(t_{abcd}+t_{abce})+n(3n-5)(t_{abde}+t_{acde}+t_{bcde})\\*
&{}+(6n^2-18n+13)t_{abcde},\\[3pt]
\mathcal F_{13}={}& -n(t_{ab}+t_{bc})-(n+1)t_{ac}+(-2n^2+2n+3)t_{abc}-n(n-2)(t_{abd}+t_{bcd}+t_{abe}+t_{bce})\\*
&{}+\bigl[n^2+(n-2)^2\bigr](t_{acd}+t_{ace})+n^2(t_{ade}+t_{bde}+t_{cde})-(n-2)(2n-5)(t_{abcd}+t_{abce})\\*
&{}+n(3n-5)(t_{abde}+t_{bcde})+(2n^2+3n-8)t_{acde}+(9n^2-28n+23)t_{abcde},\\[3pt]
\mathcal F_{14}={}& -(3n-1)t_{ab}-2n(n-2)(t_{abc}+t_{abd}+t_{abe})+n^2(t_{acd}+t_{bcd}+t_{ace}+t_{bce}+t_{ade}+t_{bde}+t_{cde})\\*
&{}+(2n^2-n-4)(t_{abcd}+t_{abce}+t_{abde})+3n(n-1)(t_{acde}+t_{bcde})+(9n^2-24n+17)t_{abcde},\\[3pt]
\mathcal F_{15}={}& n^2\sum_{\substack{S\subseteq\{a,b,c,d,e\}\\|S|=3}}t_S+3n(n-1)\sum_{\substack{S\subseteq\{a,b,c,d,e\}\\|S|=4}}t_S+(6n^2-12n+7)t_{abcde},\\[3pt]
\mathcal F_{16}={}& -n(t_{ab}+t_{ac}+t_{bd}+t_{cd})+(5-2n)(t_{abc}+t_{abd}+t_{acd}+t_{bcd})+(12n-20)t_{abcd}\\*
&{}+(-n^2+2n+1)(t_{abe}+t_{cde})+(2n^2-2n+2)(t_{ace}+t_{bde})+(n-2)^2t_{bce}+n^2t_{ade}\\*
&{}+(-2n^2+9n-12)(t_{abce}+t_{bcde})+(2n^2-n-6)(t_{abde}+t_{acde})+(4n^2-25n+32)t_{abcde},\\[3pt]
\mathcal F_{17}={}& -(n+1)(t_{ac}+t_{bd})+(4-2n)(t_{abc}+t_{abd}+t_{acd}+t_{bcd})+(10n-14)t_{abcd}\\*
&{}+(2n^2-2n+2)(t_{ace}+t_{bde})+(2n-4)(t_{abce}+t_{abde}+t_{acde}+t_{bcde})+(n^2-13n+16)t_{abcde},\\[3pt]
\mathcal F_{18}={}& t_e-(n+1)(t_{ab}+t_{cd})+(n-2)(t_{ae}+t_{be}+t_{ce}+t_{de})+(4-2n)(t_{abc}+t_{abd}+t_{acd}+t_{bcd})\\*
&{}+(-2n^2+2n+5)(t_{abe}+t_{cde})+(2n^2-4n+4)(t_{ace}+t_{bce}+t_{ade}+t_{bde})+(10n-14)t_{abcd}\\*
&{}+(8n-14)(t_{abce}+t_{abde}+t_{acde}+t_{bcde})+(9n^2-41n+42)t_{abcde}.
\end{align*}
\endgroup

\phantomsection
\bibliography{references}

\end{document}